%% file: main.tex
\documentclass[%
 aip,
 jmp,
 amsmath,amssymb,
 preprint,
]{revtex4-1}

\usepackage{graphicx}% Include figure files
\usepackage[percent]{overpic}
\usepackage{subcaption}
\usepackage{dcolumn}% Align table columns on decimal point
\usepackage{bm}
    \usepackage{amsmath}
    \usepackage{accents}
    \newlength{\dhatheight}
    \newcommand{\doublehat}[1]{%
    \settoheight{\dhatheight}{\ensuremath{\hat{#1}}}%
    \addtolength{\dhatheight}{-0.35ex}%
    \hat{\vphantom{\rule{1pt}{\dhatheight}}%
    \smash{\hat{#1}}}}
 \usepackage{amsthm}
 \newtheorem{thm}{Theorem}
 \newtheorem{corollary}{Corollary}
 \newtheorem{lemma}{Lemma}
\usepackage{url}
\usepackage{comment}
\usepackage[utf8]{inputenc}
\usepackage[T1]{fontenc}
\usepackage{mathptmx}
\DeclareMathAlphabet{\mathcal}{OMS}{cmsy}{m}{n}
\usepackage{etoolbox}
\usepackage{braket}
\usepackage{tikz}
\usetikzlibrary{quantikz}
    
\usepackage{hyperref}
\makeatletter
\def\@email#1#2{%
 \endgroup
 \patchcmd{\titleblock@produce}
  {\frontmatter@RRAPformat}
  {\frontmatter@RRAPformat{\produce@RRAP{*#1\href{mailto:#2}{#2}}}\frontmatter@RRAPformat}
  {}{}
}%
\makeatother
\begin{document}

%\preprint{AIP/123-QED}

\title[Error and Resource Estimates of VQAs for Solving DEs Based on RKMs]
{Error and Resource Estimates of Variational Quantum Algorithms for Solving Differential Equations Based on Runge-Kutta Methods}

\author{David Dechant}
\email{dechant@lorentz.leidenuniv.nl}
\affiliation{$\langle aQa^L\rangle$ Applied Quantum Algorithms Leiden, The Netherlands}

\affiliation{ 
Instituut-Lorentz, Universiteit Leiden, P.O. Box 9506, 2300 RA Leiden, The Netherlands
}

\author{Liubov Markovich}

\affiliation{$\langle aQa^L\rangle$ Applied Quantum Algorithms Leiden, The Netherlands}
\affiliation{ 
Instituut-Lorentz, Universiteit Leiden, P.O. Box 9506, 2300 RA Leiden, The Netherlands
}

\author{Vedran Dunjko}
\affiliation{$\langle aQa^L\rangle$ Applied Quantum Algorithms Leiden, The Netherlands}
\affiliation{LIACS, Universiteit Leiden, P.O. Box 9512, 2300 RA Leiden, Netherlands}

\author{Jordi Tura}
\affiliation{$\langle aQa^L\rangle$ Applied Quantum Algorithms Leiden, The Netherlands}

\affiliation{ 
Instituut-Lorentz, Universiteit Leiden, P.O. Box 9506, 2300 RA Leiden, The Netherlands
}

%\date{\today}

\begin{abstract}
A focus of recent research in quantum computing has been on developing quantum algorithms for solving differential equations using variational methods on near-term quantum devices.
A promising approach involves variational algorithms, which combine classical Runge-Kutta methods with quantum computations. However, a rigorous error analysis, essential for assessing real-world feasibility, has so far been lacking.
In this paper, we provide an extensive analysis of error sources and determine the resource requirements needed to achieve specific target errors.
In particular, we derive analytical error and resource estimates for scenarios with and without shot noise, examining shot noise in quantum measurements and truncation errors in Runge-Kutta methods. Our analysis does not take into account representation errors and hardware noise, as these are specific to the instance and the used device.
We evaluate the implications of our results by applying them to two scenarios: classically solving a $1$D ordinary differential equation and solving an option pricing linear partial differential equation with the variational algorithm, showing that the most resource-efficient methods are of order 4 and 2, respectively.
This work provides a framework for optimizing quantum resources when applying Runge-Kutta methods, enhancing their efficiency and accuracy in both solving differential equations and simulating quantum systems. 
% Further, this work plays a crucial role in assessing the suitability of these variational algorithms for fault-tolerant quantum computing.
% The results may also be of interest to the numerical analysis community as they involve the accumulation of errors in the function description, a topic that has hardly been explored even in the context of classical differential equation solvers.
\end{abstract}

\maketitle

\section{Introduction}
For more than four decades, quantum computing has captivated the minds of researchers, but significant experimental advancements have only been achieved in recent years. We are living in the era of Noisy Intermediate-Scale Quantum (NISQ) devices~\cite{jurcevic2021,henriet2020,PhysRevLett.123.170503,risinger2021characterization}, which, while capable of certain super classical computations in principle, are also susceptible to noise and errors, preventing them from providing substantial computational advantages over classical computers. Efforts to tackle these challenges have resulted in proliferate research lines dedicated to error correction and noise mitigation across various quantum architectures~\cite{PhysRevA.86.032324}. 
One approach to ameliorating some of the issues involves hybrid quantum-classical algorithms that use quantum computers to process information, while classical computers handle the correction  and optimization processes, helping to reduce errors and improve efficiency. 
For example, in variational quantum algorithms, such as the Variational Quantum Eigensolver (VQE) or the Quantum Approximate Optimization Algorithm (QAOA), the quantum part evaluates the cost function depending on a given set of parameters, while the classical part optimizes these parameters to minimize the cost function.
In physical terms, the cost function is typically the energy of a Hamiltonian that encodes the optimization task.
Extracting the final result from such a hybrid device is challenging, as errors arise both from quantum and classical sources. A notable source of error is the measurement shot noise arising out of the discrete nature of quantum measurements and the limited number of measurements taken \cite{scriva2024challenges}.

Solving differential equations (DEs) is a critical task in various scientific and engineering fields, and several quantum computing-based proposals have been developed to tackle this problem. 
First ideas\cite{leyton2008quantum,berry2014high,berry2017quantum} were built around the quantum linear system algorithm\cite{harrow2009quantum}, but they require fault-tolerant quantum computers. 
Later, approaches based on variational quantum algorithms were introduced, e.g., in Refs.~\cite{lubasch2020variational,leong2022variational,kyriienko2021solving,fontanela2021quantum}. 
In one of the latter approaches, the DE is mapped to an imaginary-time Schrödinger equation~\cite{mcardle2019variational}. 
It is solved by a variational algorithm, where the time evolution of a quantum state is approximated by a variational quantum circuit and mapped to the time evolution of the parameters of this circuit. Classically computing the parameters at the final time step, one reinserts them into the variational quantum circuit to prepare the evolved quantum state. 
This is an approach originally proposed in Ref.~\cite{li2017efficient} for quantum simulation and has attracted a lot of interest since \cite{nagano2023quench,barison2021efficient,benedetti2021hardware, yao2021adaptive, cirstoiu2020variational, heya2023subspace, pool2024nonlinear}.

In the approach of Refs.~\cite{li2017efficient, mcardle2019variational}, the resulting time evolution of the quantum circuit parameters has the form of an ordinary differential equation (ODE).
Classically, ODEs can be solved by time-discretization methods such as the Euler method or the more general Runge-Kutta methods\cite{butcher2006general} (RKMs) that can be categorized by their order $p$. 
The first proposed variational quantum
algorithm for solving differential equations based on the Euler method is shown in Ref.~\cite{lubasch2020variational}.
Those methods approximate the time evolution by calculating a truncated Taylor expansion at each time step, and the higher the order of the RKM, the lower is the resulting truncation error, but the higher is the required number of function evaluations. 
There are also different generalizations of the Euler methods, like the linear multistep methods or the general linear methods~\cite{butcher2006general}, which can lead to similar accuracies as RKMs with different requirements of the number of function evaluations. 

Previous works have proposed that choosing RKMs, such as the widely used classical RKM \cite{kutta1901beitrag} with the order $p=4$ instead of the Euler method, will also be favorable in solving the time evolution of the quantum circuit parameters\cite{li2017efficient,radha2021quantum,alghassi2022variational}. 
However, each function evaluation incorporates evaluations of quantum circuits, making higher-order RKMs more demanding on the quantum resources. 
Moreover, quantum circuit evaluations are affected by the "shot noise bottleneck," where the error scales as $\mathcal{O}(1/N_{meas}^2)$ with the number of measurements $N_{meas}$, introducing an additional source of error.
The number of measurements is a precious resource in quantum computing, as it is the most costly and time-consuming operation and as it is severely limited by the available runtime of the device before it requires recalibration~\cite{jurcevic2021,brown2016,schafer2018}.
That is why, for the practicality of the variational algorithm, it is crucial to minimize the total number of quantum circuit measurements.

In this paper, we focus on solving differential equations based on the approach of Refs.~\cite{li2017efficient,mcardle2019variational}. We investigate whether higher-order RKMs outperform the Euler method in solving the time evolution of the circuit parameters by analyzing the different sources of error and resource requirements. Specifically, we provide a detailed analysis of the total error of the variational quantum algorithm for solving differential equations, defined as the trace distance between the actual solution and the output of the variational algorithm. We explicitly consider the truncation error associated with the chosen RKM and the shot noise error. Other relevant sources of errors, such as circuit error (gate infidelity, bias, SPAM errors) and the representation error (the variational circuit being able to represent the solution at all time steps with its parameters), are assumed ideal as they depend among others on the chosen Ansatz and problem instance \cite{gacon2024variational}. 
We establish rigorous error bounds and use them to estimate the minimum number of circuit evaluations required by the algorithm for a given target error, based on the order of the RKM. 
Additionally, we perform an analysis of the RKMs under the assumption of no shot noise. 

Further, we validate the resource estimates through benchmarking: the analysis without shot noise is demonstrated with a simple ODE, while the analysis of the variational algorithm is applied to option pricing, where the dynamics are described by the Black-Scholes equation~\cite{black1973pricing,fontanela2021quantum}.  The latter is a partial differential equation that has attracted a lot of attention in the variational quantum computing community~\cite{alghassi2022variational, fontanela2021quantum, 10112619}. This shows that the application of our method is not restricted to ODEs but can also be applied to solving partial differential equations. We directly compare the total number of circuit evaluations required by the algorithm, depending on the order of the chosen RKM.
In this work, we derive rigorous error bounds that are general, but may overestimate the true error in practice. 
Our resource estimates are based on optimizing resources with respect to these upper bounds and might therefore be overly conservative. Such an approach constitutes the optimal parameter selection strategy, providing a guaranteed success probability.

Similar error and resource estimates have been conducted in Refs.~\cite{mcardle2019variational,miessen2021quantum,zoufal2023error,kolotouros2024accelerating}.
However, our analysis is integrating the truncation errors of RKMs and shot noise errors, making it comprehensive, and offers direct comparisons of resource requirements between different RKMs and an a-priori analysis that leads to substantial savings in the cost of the algorithms.

Our work is relevant not only to quantum algorithms that use variational approaches but also to those that employ RKMs for solving DEs, such as in Ref.~\cite{leyton2008quantum,zanger2021quantum}. 
Our analysis highlights the importance of studying the sensitivity of classical numerical methods for ODEs to perturbations in the input function, which has hardly been explored so far. 

From our results, we conclude that depending on the parameters distinct to the problem at hand, higher-order RKMs are decreasing the resource requirements. In the use case of option pricing via the Black-Scholes equation, we showed that an RKM of order $p=2$ is requiring the minimal number of total quantum circuit evaluations. 
For other applications, even higher-order methods might be favorable. 
With our thorough analysis of the involved parameters, we provide a straightforward framework that can be applied to other use cases that can be tackled by solving a DE in the form of Eq.~\eqref{eq: generalized imaginary time schroedinger equation} and decrease the resource requirements of the variational algorithms by suggesting the most efficient RKM.

In the interest of making the paper self-contained, we build up the paper in the following way:
We begin with an introduction to the RKM and the variational quantum algorithm from ~\cite{mcardle2019variational} and with the problem statement. 
In Sec. \ref{sec: analysis classical ode solver}-\ref{sec: Analysis with Shot Noise Scaling}, we show estimates of the errors and minimal resources required for ODE solving without and with shot noise. 
In Sec.~\ref{sec:Introduce Quantum Circuits, McLachlan Principle, Shot noise, Lipschitz Number}, we analyze parameters of the variational algorithm that the error and resource estimates depend on. 
Afterwards, we are performing numerical analyzes of a simple ODE without shot noise and of an option pricing use case in Sec.~\ref{sec: numerics}. 
In Sec.~\ref{sec:conclusion}, we provide a discussion of the results and conclusions.

\section{Preliminaries}
\label{sec:preliminaries}
In this article, we investigate the impact of different types of errors on variational quantum algorithms for solving DEs based on Runge-Kutta methods. These algorithms are motivated by the variational quantum algorithm for imaginary time evolution introduced in Ref.~\cite{mcardle2019variational}. In the following, we firstly present the Runge-Kutta methods, which is a family of classical methods for solving ODEs. Secondly, we present the variational quantum algorithm for solving DEs that are based on the Runge-Kutta methods. And thirdly, we introduce the errors and resources that we analyze in this work.

\subsection{Runge-Kutta Methods}
\label{sec:Preliminaries: Runge-Kutta Methods}
Let us consider the initial value problem, which is an ODE together with an initial condition:
\begin{align}\label{eq: initial value problem}
    \frac{d y(\tau)}{d\tau}&= f\left(\tau,y(\tau)\right)\\
    y(\tau_0)&=y_0\ ,\nonumber
\end{align}
where $\tau$ is the time and $y(\tau)$ is an element of the image of an unknown function in a scalar or vector form that we want to determine, and where $f\left(\tau,y(\tau)\right)$ fulfills the assumptions of the Picard-Lindel\"{o}ff theorem, guaranteeing the existence of a unique solution to Eq.~\eqref{eq: initial value problem}.

Since most of the time an analytically closed form is not viable, numerical methods are the only way to obtain an approximate solution. A common way of solving Eq.~\eqref{eq: initial value problem} is with the so-called Runge-Kutta methods (RKMs).
The RKMs are a class of methods based on the Taylor expansion of $y$ in order to approximate the numerical solution of the ODE at the future time step by using evaluations of $f\left(\tau,y(\tau)\right)$.
\par 
A general blueprint of RKMs with $s$ stages can be outlined as follows: For simplicity, let us set $\tau_0=0$. We divide the time interval $\tau\in[0,T]$, $T>0$ into $N_{\tau}$ time steps  denoted  by $\tau_n$, $n\in [1,N_{\tau}]$. We assume the step size $\Delta \tau = \tau_n - \tau_{n-1}=T/N_{\tau}$  to be constant and denote the computed solution at the $n$-th time step by $y(\tau_n)$. 
    Using  Eq.~\eqref{eq: initial value problem} and the solution $y(\tau_n)$ at the $n$-th time step, we can compute $y(\tau_{n+1})$ in the following way (see Ref.~\onlinecite[p.907]{press2007numerical}):
    \begin{align}
    \label{eq:general RK equation}
    y(\tau_{n+1}) = y(\tau_n) + \Delta \tau \sum_{i=1}^{s} b_i k_i\ ,
    \end{align}
    where the calculation of the latter function is done in $s$ stages
    \begin{align}\label{1619}
     k_1 & = f(\tau_n, y(\tau_n)), \\\nonumber
     k_2 & = f(\tau_n+c_2\Delta \tau, y(\tau_n)+a_{21}k_1\Delta \tau), \\\nonumber
     k_3 & = f(\tau_n+c_3\Delta \tau, y(\tau_n)+(a_{31}k_1+a_{32}k_2)\Delta \tau), \\\nonumber
         & \ \ \vdots \\\nonumber
     k_{s} & = f(\tau_n+c_{s}\Delta \tau, y(\tau_n)+(a_{{s}1}k_1+a_{{s}2}k_2+\cdots+a_{{s},{s}-1}k_{{s}-1})\Delta \tau).
    \end{align}
    The constants $a_{i,j}$ ($1\leq j< i\leq {s}$), $b_i$ ($1\leq i\leq {s}$ ) and $c_i$ ($2\leq i\leq {s}$ ) are specific for each RKM. In order to be consistent, the constants have to satisfy
    \begin{align}
    \label{eq:RK, condition on b_i}
        \sum_{i=1}^{s}b_i=1,\quad \text{and }\quad
        \sum_{j=1}^{i-1}a_{i,j}=c_i,\quad \text{for}  \quad 2\leq i\leq {s}.
    \end{align}
    For later analysis, we define the maxima of these parameters for a specific RKM in the following way:
    \begin{align}
    \label{eq:definition_bmax}
    b_{max}=\max_i|b_i|,\\
    a_{max}=\max_{i,j}|a_{i,j}|\ .\label{eq:definition_amax}
    \end{align}
    
    The simplest  RKM is the Euler method ($s=1$), which approximates the function in one stage iteratively as follows:
    \begin{align}
    \label{eq:euler method}
        y(\tau_i+\Delta \tau)=y(\tau_i)+\Delta \tau f\left(\tau_i,y(\tau_i)\right).
    \end{align}
   The estimation error $\ell_i$ induced at each time step $\tau_i$ due to the truncation of the Taylor series is referred to as the local truncation error (LTE) of the method. RKMs are classified according to the error scaling of their LTE. A RKM is said to have an order $p$ if the LTE is bounded by an error that scales as $\mathcal{O}(\Delta \tau^{p+1})$. Generally, the order and the number of stages of an RKM are related by $s=p$ for $1\leq p\leq 4$, and $s>p$ for $p\geq 5$. This discrepancy is due to the fact that finding the coefficients $a_{i,j},b_{i}$ and $c_{i}$ becomes increasingly difficult for higher values of $p$ as it involves solving a system of non-linear equations that grows more complicated for higher $p$. For higher-order methods this can only be achieved with an increasingly higher number of stages $s$. To the best of our knowledge, there is no closed formula for calculating the minimum number of stages required for a specific order. The relations up to order $p=10$ are provided in Table~\ref{tab:Relation between order and number of stages of Runge-Kutta methods}.

    \begin{table}[htbp]
    \centering
    \begin{tabular}{|c|c|}
        \hline
        \textbf{Order} $p$ & \textbf{Number of stages} $s$ \\
        \hline
        \( 5 \) & \( 6 \) \\
        \hline
        \( 6 \) & \( 7 \) \\
        \hline
        \( 7 \) & \( 9 \) \\
        \hline
        \( 8 \) & \( 11 \) \\
        \hline
        \( 9 \) & \( 13 \) \\
        \hline
        \( 10 \) & \( 16 \) \\
        \hline
    \end{tabular}
    \caption{Relation between order and the minimum number of stages of RKMs~\cite{butcher2006general, khashin2009symbolic,zhang2024explicit}
    }
    \label{tab:Relation between order and number of stages of Runge-Kutta methods}
\end{table}
The following theorem provides an upper bound on the LTE of a $p$-th order RKM:
\begin{thm}(See Ref.~\onlinecite[p.180]{butcher2016numerical})
    \label{thm:local_truncation_RKE}
    The LTE $\ell_n$ of the  $p$-th order RKM  at the step $n\in[1,N_{\tau}]$ is bounded by
    \begin{align}
        \|\ell_n\|=\|y(\tau_n)-y_n\|\leq \Delta \tau^{p+1}  K  L^p_{f\tau}   M + \mathcal{O}(\Delta \tau^{p+2})\ ,
    \end{align}
    where $\|\cdot\|$ denotes a norm, which can be any norm on the state space (e.g., Euclidean, maximum, or 1-norm), as the bound is independent of the specific choice.
    Here, $y_n$ is the  RKM approximation of $y(\tau_n)$ calculated at step $n$, assumed that the value $y(\tau_{n-1})$ is exact. Further, $K>0$ is a constant depending on the chosen  RKM and $f(\tau_n,y(\tau_n))$ can be upper bounded by $M>0$ and is Lipschitz-continuous with respect to the first variable, such that the following bounds hold:
      \begin{align}
        \left\|f(\tau_n,y(\tau_n))\right\|\leq M,\quad 
        \left\|\frac{\partial f(\tau_n,y(\tau_n))}{\partial \tau} \right\|\leq L_{f\tau},\dots,\quad \left\|\frac{\partial^i f(\tau_n,y(\tau_n))}{\partial \tau^i}\right\|\leq \frac{L_{f\tau}^i}{M^{i-1}}\ ,\label{eq:boundL}
    \end{align}
    for all $1\leq i\leq p$, $0<\tau_n<T$ and $f(\tau_n, y(\tau_n)):=\frac{\partial y(\tau)}{\partial \tau}\Big{|}_{\tau_n}$.
        \end{thm}
Unless stated otherwise, the same norm is used consistently throughout the analysis for states, errors, Lipschitz bounds, and the corresponding induced operator norms.
Since all norms on finite-dimensional spaces are equivalent, the order of the error bounds remains unaffected by the specific choice of norm, with only the constants differing.
In our setting, we use the trace norm for quantum states and the Euclidean (2-)norm, together with its induced operator norm, for all other quantities.

We are estimating the constants $K,M$ and $L_{f\tau}$ in Sections~\ref{subsubsection: Parameter estimation and sensitivity classical} and~\ref{subsubsec: Parameter estimates and Numerical Analysis quantum} for the specific examples of DEs that we cover numerically in Section~\ref{sec: numerics}.

\subsection{Variational Quantum Algorithm for Solving Differential Equations}
\label{sec: Variational Quantum Simulator preliminaries}
In this section, we present the variational algorithm for solving linear differential equations of the following type: 
\begin{align}\label{eq: generalized imaginary time schroedinger equation}
    \frac{d y(\tau)}{d\tau}&= -\mathcal{H}\cdot y(\tau)\\
    y(\tau_0)&=y_0\ , 
\end{align}
where $y(\tau)\in \mathbb{C}^d$ is a vector and $\mathcal{H}$ a $d\times d$-dimensional matrix. We assume w.l.o.g. that $d$ is a power of $2$ and that $\mathcal{H}$ is Hermitian. This differential equation is a special case of the initial value problem as stated in Eq.~\eqref{eq: initial value problem}.

It is possible to map a wide variety of DEs to Eq.~\eqref{eq: generalized imaginary time schroedinger equation}, such as stochastic DEs (see Ref.~\cite{alghassi2022variational}), the Black-Scholes partial DE (see our analysis in  Sec.~\ref{subsec:option pricing} and Ref.~\cite{fontanela2021quantum}), or other linear partial DEs (see Ref.~\cite{jin2022quantum}).
Typically, this mapping involves a discretization of the underlying space and differential operators onto a grid~\cite{alghassi2022variational,fontanela2021quantum,jin2022quantum}, but it is also possible to encode the solution via spectral methods like Chebyshev polynomials or in Fourier basis (see for example Ref.~\cite{paine2023physics}, although they use a different quantum algorithm). 

If $d$ is not a power of $2$, it is always possible to embed $\mathcal{H}$ and $y(\tau)$ into higher-dimensional spaces, such that the assumption holds. 
The resulting higher dimensional space into which $y(\tau)$ is embedded is less than a factor of $2$ larger than the original space.

If $\mathcal{H}$ is not Hermitian, several strategies can be applied: For some differential equations, it is possible to apply changes of variables in order to transform $\mathcal{H}$ to a Hermitian matrix, such as the transformation done in Ref.~\cite{fontanela2021quantum}. Alternatively, it is possible to divide any matrix $\mathcal{H}$ into a Hermitian and an anti-Hermitian part, which effectively leads to a doubling in the number of circuit evaluations, as demonstrated in Ref.~\cite{sokolov2023orders}. 
It is also possible to use the technique shown in Ref.~\cite{jin2024quantum} that gives a mapping from Eq.~\eqref{eq: generalized imaginary time schroedinger equation} with a non-Hermitian $\mathcal{H}$ to the real time Schrödinger equation, which one can solve with the variational algorithm introduced in Ref.~\cite{li2017efficient}.
Further, there exists a generalization of the variational algorithm for imaginary time evolution that we present here which can be applied to linear differential equations with non-Hermitian matrices $\mathcal{H}$ without the need to embed them first into higher-dimensional Hilbert spaces \cite{endo2020variational}. Note that our analysis can easily be adapted to this generalized variational algorithm, as well as to the variational algorithm solving the real time Schrödinger equation.

The variational quantum algorithm that solves Eq.~\eqref{eq: generalized imaginary time schroedinger equation}, is based on the variational quantum algorithm for imaginary time evolution introduced in Ref.~\cite{mcardle2019variational}. Therefore, we firstly bring Eq.~\eqref{eq: generalized imaginary time schroedinger equation} into the form of quantum imaginary time evolution.

For the remainder of the paper we will be assuming $\mathcal{H}$ is Hermitian. The matrix $\mathcal{H}$ can thus be decomposed in the following way:
    \begin{align}
    \mathcal{H}=\sum_{m=1}^{N}\lambda_m\sigma_m \ .\label{eq: hamiltonian}
    \end{align}
    Here, $\{\lambda_m\}_{m=1}^{N}\in \mathbb{R}$ are the decomposition coefficients and $\{\sigma_m\}_{m=1}^{N}$ are the Pauli strings, which are tensor products of single qubit Pauli matrices and the identity.
   This  possible decomposition  always exists and is unique, since the collection of $d^2$ Pauli strings form an orthogonal basis for Hermitian operators acting on a $d$ dimensional Hilbert space. The operator $\mathcal{H}$ plays the role of the Hamiltonian of a quantum system. Effectively replacing the real time parameter $t$ from the Schrödinger equation with $-i\tau$, where $\tau$ is a real number, Eq.~\eqref{eq: generalized imaginary time schroedinger equation} can then be considered as a Schr\"odinger equation in imaginary time.
\par
Let us further realize the function $y(\tau)$ as a quantum state $\ket{\psi(\tau)}\in \mathbb{C}^d$ and the initial condition $y(\tau_0)=y_0$ as a state $\ket{\psi(0)}\in \mathbb{C}^d$. That can be done by normalizing the entries of the vector $y(\tau)$ and by subsequently encoding the resulting entries as the amplitudes of $\ket{\psi(\tau)}$:
\begin{align}\label{eq:definition_of_psi}
\ket{\psi(\tau)}=\sum_{i=0}^{d-1}  \frac{y(\tau)_i}{\sqrt{\|y(\tau)\|_2^2}}\ket{i}\ ,
\end{align}
where $\{\ket{i}\}_{i=0}^{d-1}$ is the computational basis. 

The problem of solving Eq.~\eqref{eq: generalized imaginary time schroedinger equation} then comes down to simulating the corresponding non-unitary time evolution $V(\tau)=e^{-\mathcal{H}\tau}$ that solves the following equation:
\begin{align}\label{eq: generalized imaginary time schroedinger equation, ket formulation}
    \frac{d \ket{\psi(\tau)}}{d\tau}&= -\mathcal{H}\cdot \ket{\psi(\tau)}
\end{align}
with the initial state $\ket{\psi(0)}$ at time $\tau=0$.
The goal is then to calculate the evolved state $\ket{\psi(T)}$ at  time $T$, which can be calculated from the initial state as:
     \begin{align}\label{1240}
         \ket{\psi(T)}=\gamma(T)V(T)\ket{\psi(0)},\quad \text{ with the normalization  }\quad \gamma(\tau)=\left(\bra{\psi(0)}V(2T)\ket{\psi(0)}\right)^{-1/2}\ .
     \end{align} 
While reading out the amplitudes of  $\ket{\psi (T)}$ itself will only give the ratio of each basis vector in the solution $y(T)$ of the DE in Eq.~\eqref{eq: generalized imaginary time schroedinger equation}, keeping track of the initial normalization from the mapping of $y(\tau_0)$ to $\ket{\psi (0)}$ and the intermediate normalization factors $\gamma(\tau)$ will give the resulting renormalization factor. One can identify $\gamma(\tau)$ as one of the parameters that is updated at each step (see for example Refs.~\cite{lubasch2020variational, alghassi2022variational}).

\par As proposed in Ref.~\cite{arute2019quantum}, the evolution of the state $\ket{\psi(\tau)}$ can be simulated by using a parameterized quantum circuit to approximate the evolved state. Instead of $\ket{\psi(\tau)}$, the  variational quantum "trial" state 
 \begin{align}
        \label{eq:Ansatz}
    \ket{\phi(\bm{\theta}(\tau))}:= R(\bm{\theta}(\tau))\ket{\overline{0}},\quad \ket{\phi(\bm{\theta}(\tau))}\in \mathbb{C}^d,
    \end{align}
    prepared by a variational circuit $R(\bm{\theta}(\tau))$ with a  vector  of  time-dependent parameters $\bm{\theta}(\tau)=(\theta_{1}(\tau),\theta_{2}(\tau),\dots, \theta_{N_V}(\tau))^T\in \mathbb{R}^{N_V}$
 is taken. The trial state is often referred to as the \textit{Ansatz}.
 The circuit is chosen in  such a way that it consists of a sequence of $N_V$ layers each depending on one variational parameter as follows:
    \begin{align}
        \label{eq:circuit on total state}
        R(\bm{\theta}(\tau))=R_{N_V}(\theta_{N_V}(\tau))R_{N_V-1}(\theta_{N_V-1}(\tau))\dots R_{1}(\theta_1(\tau))\ .
    \end{align}
    Each $R_k(\theta_k(\tau))$, $k\in (1,\dots, N_V)$ is a unitary operator that can be written as
    \begin{align}
        \label{eq:circuit on one qubit}R_k(\theta_k(\tau))=\exp{\left(\theta_k(\tau)\sum\limits_{i=1}^{N_d}f_{k, i}\sigma_{k, i}\right)}\ ,
    \end{align}
    with fixed complex parameters $f_{k,i}$ and Pauli strings $\sigma_{k, i}$.
    For simplicity, we keep $N_d$ fixed.  In general, only a small subspace of the Hilbert space can be reached with such an Ansatz, but as was shown in Ref.~\cite{PhysRevLett.106.170501}, this is sufficient for physically relevant states. Furthermore, this Ansatz captures a wide class of possible implementations, such as the coupled cluster ansatz \cite{taube2006new} or hardware efficient methods \cite{dallaire2019low, peruzzo2014variational}.
    
    \par The idea of the method is to map the time evolution of the state $\ket{\psi(\tau)}$ according to the Hamiltonian $\mathcal{H}$ to a time evolution of the parameters $\bm{\theta}(\tau)$ of the state  $\ket{\phi(\bm{\theta}(\tau))}$.
    To this end, one first finds the vector of parameters 
    $\bm{\theta}(\tau=0)$ such that it minimizes the distance $\|\ket{\phi(\bm{\theta}(0))}-\ket{\psi(0)}\|$.
The  McLachlan's variational principle~\cite{mclachlan1964variational,yuan2019theory}as given by
    \begin{align}\label{eq:mclachlanfirstmention}
        \delta\|(d/d\tau + \mathcal{H})\ket{\phi(\bm{\theta}(\tau))}\|_1=0
    \end{align} 
    is fulfilled if Eq.~\eqref{eq: generalized imaginary time schroedinger equation, ket formulation} is valid for the trial state $\ket{\phi(\bm{\theta}(\tau))}$.
    Here $\|\ket{\phi}\|_1=\sqrt{\langle \phi|\phi\rangle}$ is the trace norm of a quantum state $\ket{\phi}$ and $\delta$ denotes infinitesimal variation. 
    However, if the chosen Ansatz is not expressive enough or too biased, Eq.~\eqref{eq:mclachlanfirstmention} will not hold. 
    The resulting errors are hard to control in practice, but there exist ways to estimate them. See for example Ref.~\cite{zoufal2023error,gacon2024variational}.
    
    Eq.~\eqref{eq:mclachlanfirstmention} is used to translate the time evolution of the trial state  $\ket{\phi(\bm{\theta}(\tau))}$ to a time evolution of the vector of parameters $\bm{\theta}(\tau)$, given as an ODE (see Appendix \ref{appendix: McLachlan's principle}):
     \begin{align}
        \label{eq: parameter ODE}
        \sum\limits_{l=1}^{N_V}A_{kl}\left(\bm{\theta}(\tau)\right)\frac{\partial \theta_l(\tau)}{\partial \tau}=C_{k}\left(\bm{\theta}(\tau)\right)\quad \forall\tau\ .
    \end{align}
Note that this ODE is acting on the $N_V$-dimensional vector $\bm{\theta}(\tau)$, where $N_V$ is independent of the dimension $d$ of the time evolution in Eq.~\eqref{eq: generalized imaginary time schroedinger equation, ket formulation}, and is entirely determined by the number of parameters of the chosen Ansatz. In particular, $N_V$ is not bound to be a power of $2$. In principle, lower $N_V$ will reduce the total number of circuit evaluations for our algorithm (see Sec.~\ref{sec: Variational Quantum Simulator}), but will also lead to a lower expressivity of the corresponding Ansatz.

The elements of the matrix $A$ and the vector $C$ are: 
\begin{align}\label{2144}
        A_{kl}\left(\bm{\theta}(\tau)\right)&=\mathfrak{Re}\left(\frac{\partial\bra{\phi(\bm{\theta}(\tau))}}{\partial\theta_k}\frac{\partial\ket{\phi(\bm{\theta}(\tau))}}{\partial\theta_l}\right), \\ 
        C_k\left(\bm{\theta}(\tau)\right)&=\mathfrak{Re}\left(-\frac{\partial\bra{\phi(\bm{\theta}(\tau))}}{\partial\theta_k}{\mathcal{H}}\ket{\phi(\bm{\theta}(\tau))}\right).
        \label{2144_1}
    \end{align}
     \par Taking the derivatives of  Eq.~\eqref{eq:circuit on one qubit} with respect to the parameters, one can calculate the derivative of the trial state in Eq.~\eqref{eq:Ansatz}:
    \begin{align}
        \frac{\partial \ket{\phi(\bm{\theta}(\tau))}}{\partial \theta_k}=\sum\limits_{i=1}^{N_d}f_{k,i}R_{k,i}\ket{\overline{0}},\quad R_{k,i}=R_{N_V}R_{N_V-1}...R_{k+1}R_k\sigma_{k,i}R_{k-1}\dots R_2R_1 ,\label{eq: Ansatz derivative}
    \end{align}
    where we omitted the dependencies of the $R_k$s on $\theta_k$s for simplicity.
     With the  chosen \textit{Ansatz} in Eq.~\eqref{eq: Ansatz derivative}, the matrix elements in Eq.~\eqref{2144} and in Eq.~\eqref{2144_1} can be computed as:
    \begin{align}
    \label{eq: A evaluation}
    A_{k,l}\left(\bm{\theta}(\tau)\right)&=\sum\limits_{i,j=1}^{N_d}\left(f^*_{k,i}f_{l,j}\bra{\overline{0}}R^{\dagger}_{k,i}R_{l,j}\ket{\overline{0}}+h.c.\right) ,\\
    \label{eq: C evaluation}
    C_k\left(\bm{\theta}(\tau)\right)&=\sum\limits_{i=1}^{N_d}\sum\limits_{m=1}^{N}\left(f^*_{k,i}\lambda_m\bra{\overline{0}}R_{k,i}^{\dagger}\sigma_mR\ket{\overline{0}}+h.c.\right)  ,
    \end{align}
    by measuring circuits illustrated in Fig.~\ref{fig:circuit A}. Alternatively, it is possible to calculate the matrix element with parameter-shift rules \cite{mitarai2018quantum,markovich2024parameter}, making the circuits shorter in depth and therefore potentially more suitable for NISQ applications. They are applied to similar algorithms in Ref.~\cite{anuar2024operatorprojectedvariationalquantumimaginary, kolotouros2024accelerating}.
\begin{figure}[ht]
        \centering
        \scalebox{1.0}{\input{circuitAinTikZ.tex}}
        \caption{The quantum circuit evaluating the elements of $A$ and $C$ as given in Eqs.~\eqref{eq: A evaluation} and~\eqref{eq: C evaluation}. The controlled unitary $U_{k,i}$ is one of the $\sigma_{k,i}$. Depending on if one is evaluating $A$ or $C$, the controlled unitary $U_{l,j}$ is another $\sigma_{l,j}$ or one of the Pauli strings $\sigma_m$ (in which case we take $l=N_V+1$) that constitute the Hamiltonian, respectively .}
        \label{fig:circuit A}
    \end{figure}
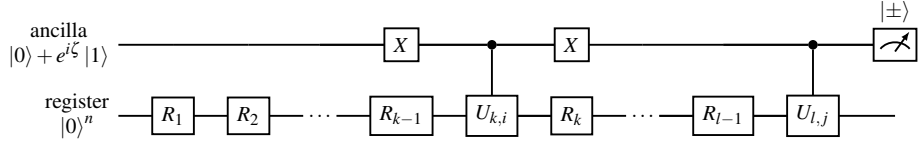

If the matrix $A\left(\bm{\theta}(\tau)\right)$ is invertible for all $\bm{\theta}(\tau)$ in the relevant domain, then Eq.~\eqref{eq: parameter ODE} can be written in the form of Eq.~\eqref{eq: initial value problem} by identifying $y(\tau)$ with $\bm{\theta}(\tau)$: 
   \begin{align}\label{eq: parameter initial value problem}
       \frac{\partial \bm{\theta}(\tau)}{\partial \tau} = f\left(\bm{\theta}(\tau)\right):= A^{-1}\left(\bm{\theta}(\tau)\right)C\left(\bm{\theta}(\tau)\right),
   \end{align}
and with the initial condition $\bm{\theta}(0)$ at time $\tau=0$.
The time evolution of the parameters $\bm{\theta}(\tau)$ is given as this ODE, and solving it until final time $T$ gives the vector of parameters $\bm{\theta}(T)$ that yield the final state $\ket{\phi(\bm{\theta}(T))}$, which serves as an approximation of the evolved state $\ket{\psi(T)}$.
If $A$ is not invertible, regularization introduces an additional error (see Sections~\ref{subsec: Total error and resources} and~\ref{subsec: Estimation of the Shot noise}). 
Similarly, the inversion itself will inevitably also include an additional error that has to be taken into account. 
As long as the errors stemming from regularization and inversion of $A$ fulfill the error bounds in this work (such as Eq.~\eqref{eq:upperboundfromliubovsensitivity}), our estimates are still applicable.
We assume therefore for the remainder of this work that $A$ is invertible.

We can thus solve a DE of the form in Eq.~\eqref{eq: generalized imaginary time schroedinger equation} with a hybrid quantum-classical algorithm by solving the ODE in Eq.~\eqref{eq: parameter ODE} defining the parameters $\bm{\theta}(\tau)$ of an Ansatz state. 

In the next section, we introduce the errors and resources required for solving the DEs in Eqs.~\eqref{eq: initial value problem} and~\eqref{eq: generalized imaginary time schroedinger equation} with the Runge-Kutta methods and with the variational quantum algorithm, respectively.

\subsection{Errors and Resources}
\label{subsec: Total error and resources}
In this paper, we analyze errors and resource requirements of the Runge-Kutta methods and the variational algorithm as described above. For this, we focus on analyzing the errors involved in solving the ODE using the RKM described in Sec.~\ref{sec:Preliminaries: Runge-Kutta Methods},  and for the variational algorithm additionally the total error arising from preparing the trial state $\ket{\phi(\bm{\theta}(T))}$ defined in Eq.~\eqref{eq:Ansatz} using the method described in Sec.~\ref{sec: Variational Quantum Simulator preliminaries}. 
In our work, we are not considering representation errors stemming from the chosen Ansatz state not being expressive enough, as it proves very challenging to estimate these errors in general. 
However, they have an influence on both errors that we estimate, by playing a role in the approximation of the time evolution of $\bm{\theta}(\tau)$ and in the approximation of the final state $\ket{\psi(T)}$.
They are specific to the task at hand, it is possible to derive a posteriori error bounds and they become less relevant for deep Ans\"{a}tze, see for example Refs.~\cite{zoufal2023error,gacon2024variational}. Also, we assume circuit error such as gate infidelity,
bias and SPAM errors to be negligible. 
We further assume the matrices $A$ as defined in Eq.~\eqref{eq: A evaluation} to be invertible (see our discussion in Sec.~\ref{sec:Introduce Quantum Circuits, McLachlan Principle, Shot noise, Lipschitz Number}). In practice, it might be necessary to introduce matrix regularizations for the cases where $A$ is not invertible (see techniques in Refs.~\cite{alghassi2022variational,mcardle2019variational,fontanela2021quantum, anuar2024operatorprojectedvariationalquantumimaginary}), which would lead to additional errors. 
\par We denote the total error arising while approximating the solution of the ODE (Eq.~\eqref{eq: initial value problem}) for  noiseless evaluations of the functions $f(\tau,y(\tau))$:
\begin{align}\label{eq: initial value problem noiseless}
\epsilon_{ODE}^{(0)}:=
\|y(\tau_{N_{\tau}})-y_{N_{\tau}}\|_{2}\ ,
\end{align}
and analyze it in Sec.~\ref{sec: analysis classical ode solver} (Theorem~\ref{thm:classicalODEerror}).
For the variational quantum algorithm, the parameters calculated via an RKM from Eq.~\eqref{eq: parameter initial value problem} additionally are influenced by shot noise in the evaluations of $f(\bm{\theta}(\tau))$, which we denote by the superscript $\delta$:
\begin{align}\label{eq: initial value problem noisy}
\epsilon_{ODE}^{(\delta)}:=
\|\bm{\theta}(\tau_{N_{\tau}})-\hat{\bm{\theta}}_{N_{\tau}}\|_{2}\ .
\end{align}
We added the hat to the approximation of $\bm{\theta}_{N_{\tau}}$ of the parameters calculated via an RKM, in order to show that they may be perturbed (e.g., from shot noise).
We analyze $\epsilon_{ODE}^{(\delta)}$ in Sec.~\ref{sec: Analysis with Shot Noise Scaling}.

Whenever the norm is not specified, we are using the $2$-norm:
\begin{align}
    \|\bm{\theta}(\tau_{N_{\tau}})-\hat{\bm{\theta}}_{N_{\tau}}\|_{2}:=\left(\sum_{i=1}^{N_V}|\bm{\theta}(\tau_{N_{\tau}})_i-\hat{\bm{\theta}}_{N_{\tau},i}|^2\right)^{1/2}\ ,
\end{align}
where $N_V$ is the length of the vectors $\bm{\theta}(\tau_{N_{\tau}})$ and $\hat{\bm{\theta}}_{N_{\tau}}$.

\par We define the total error arising from applying the variational algorithm to approximate the state $\ket{\psi(T)}$ with the trial state $\ket{\phi(\hat{\bm{\theta}}_{N_{\tau}})}$ as:
\begin{align}\label{eq:total error definition}
    \epsilon_{total}:=\|\psi(T)-\phi(\hat{\bm{\theta}}_{N_{\tau}})\|_1\ ,
\end{align}
where $\psi(T)=\ket{\psi(T)}\bra{\psi(T)}$ and $\phi(\hat{\bm{\theta}}_{N_{\tau}})=\ket{\phi(\hat{\bm{\theta}}_{N_{\tau}})}\bra{\phi(\hat{\bm{\theta}}_{N_{\tau}})}$, where $\ket{\psi(T)}$ and $\ket{\phi(\hat{\bm{\theta}}_{N_{\tau}})}$ are as defined in Eqs.~\eqref{eq:definition_of_psi} and~\eqref{eq:Ansatz}. For $\epsilon_{total}$, we use the trace distance as the most convenient distance for quantum states, which for two pure states is defined as:
\begin{align}
    \|\psi(T)-\phi(\hat{\bm{\theta}}_{N_{\tau}})\|_1:=\sqrt{1-|\braket{\psi(T)|\phi(\hat{\bm{\theta}}_{N_{\tau}})}|^2} \ .
\end{align}
We analyze the error $\epsilon_{total}$ in Sec.~\ref{sec: Variational Quantum Simulator}.

\par Based on the error estimates, we are estimating the resources needed in order for executing the RKMs and the variational algorithm.
We define the cost of solving an ODE with an RKM as the total number of times that the function $f(\bm{\theta}(\tau))$ has to be evaluated for the whole RKM:
\begin{align}\label{eq:cost function with shot noise}
    C(N_{\tau},N_r,s,p):=sN_{\tau}(s,p)N_r(s,p)\ ,
\end{align}
where  $N_r$ is the number of measurements of the function $f(\bm{\theta}(\tau))$ at a single stage of one RKM time step. In the absence of shot noise, the cost in Eq.~\eqref{eq:cost function with shot noise} reduces to
\begin{align}\label{eq:cost function no shot noise}
    C(N_{\tau},s,p):=sN_{\tau}(s,p)\ ,
\end{align}
as at each stage and time step, one needs only one evaluation of $f(\tau,y(\tau))$.
\par We are estimating the minima of the costs in Eq.~\eqref{eq:cost function with shot noise} and in Eq.~\eqref{eq:cost function no shot noise} required to reach a specified target error $\epsilon_{target}^{(\delta)}$ that upper bounds the error $\epsilon_{ODE}^{(\delta)}\leq \epsilon_{target}^{(\delta)}$  in Sec.~\ref{sec: analysis classical ode solver} and~\ref{sec: Analysis with Shot Noise Scaling}. 
\par For the variational algorithm described in Sec.~\ref{sec: Variational Quantum Simulator preliminaries}, we also determine the total number of quantum circuit evaluations required. As we will demonstrate, calculating $f(\bm{\theta}(\tau))$ for a given input $\bm{\theta}(\tau)$ requires the preparation and measurement of several quantum circuits. This significantly increases the total number of quantum circuit evaluations beyond the cost in Eq.~\eqref{eq:cost function with shot noise}. We refer to this as $ N_{circ}$, and we estimate it in Sec.~\ref{sec: Variational Quantum Simulator}.

\section{Runge-Kutta methods without Shot Noise}
\label{sec: analysis classical ode solver}
For ease of exposition, in this section we analyze the total error arising from solving an ODE with an RKM without the presence of shot noise in the evaluations of $f(\tau,y(\tau))$ in Eq.~\eqref{eq: initial value problem}. 
We also determine the minimal number $N_{\tau}^{(0)}$ of RKM steps required to achieve a prescribed accuracy.

In Theorem~\ref{thm:classicalODEerror}, we establish an upper bound on the total error, denoted $\epsilon_{ODE}^{(0)}$ in Eq.~\eqref{eq: initial value problem noiseless}, by analyzing the error propagation due to the truncation error of the RKM (see Theorem~\ref{thm:local_truncation_RKE}). 

We select the upper bound of the error $\epsilon_{ODE}^{(0)}$ as a target error, meaning the maximal error we can expect. 
This target error is the used to determine the minimal number of RKM steps needed to ensure that the resulting error remains within the target. 
The result is formalized in Theorem~\ref{thm:classicalODEresources}, and the corresponding minimal cost, as defined in Eq.~\eqref{eq:cost function no shot noise}, is derived in Corollary~\ref{cor:classicalODEcost}.

\begin{thm}\label{thm:classicalODEerror}
    Let $y(\tau_{N_{\tau}})$ be the solution of Eq.~\eqref{eq: initial value problem}. We assume that the assumptions of Theorem~\ref{thm:local_truncation_RKE} hold. 
Let us further assume that there exists a Lipschitz constant $L_{fy}$, such that $\forall y_1(\tau_n),y_2(\tau_n)$ in the spaces $\{y(\tau_n):\tau_n\in [0,T]\}$ and $\{y_{n}:n\in [1,N_{\tau}]\}$, the following holds:
\begin{align}\label{eq:definition Lipschitz constant}
    \left\|f(\tau_n,y_1(\tau_n))-f(\tau_n,y_2(\tau_n))\right\|
    \leq & L_{fy}\|y_1(\tau_n)-y_2(\tau_n)\|\ , \forall\tau_n\in [0,T].
\end{align}
    Then, the approximation $y_{N_{\tau}}$ of the solution $y(\tau_{N_{\tau}})$ at time $\tau_{N_{\tau}}$ calculated via a $p$-th order RKM with $s$ stages in $N_{\tau}$ time steps has the  error defined in Eq.~\eqref{eq: initial value problem noiseless} bounded by:
\begin{align}\label{eq:error_bound_epsilon0}
    \epsilon_{ODE}^{(0)}\leq  \frac{(1+ F(N_{\tau},s))^{N_{\tau}}-1}{F(N_{\tau},s)}\left(\frac{T}{N_{\tau}}\right)^{p+1}  K  L^p_{f\tau}   M\ ,
\end{align}
where we used the notation
\begin{align}\label{eq:def of function F}
    F(N_{\tau},s) := 
        \frac{b_{max}}{a_{max}} \left(\left(1+L_{fy}a_{max} \frac{T}{N_{\tau}}\right)^{s}-1\right) .
\end{align}
\end{thm}
For a definition of $a_{max}$ and $b_{max}$ see Eqs.~\eqref{eq:definition_bmax} and~\eqref{eq:definition_amax}.
We define the upper bound in Eq.~\eqref{eq:error_bound_epsilon0} as the target error, i.e., the theoretical maximum error that can occur when applying this method:
   \begin{align}\label{eq:definition target error noiseless}
       \epsilon_{target}^{(0)}:=  \frac{(1+ F(N_{\tau},s))^{N_{\tau}}-1}{F(N_{\tau},s)}\left(\frac{T}{N_{\tau}}\right)^{p+1}  K  L^p_{f\tau}   M\ .
   \end{align}
Based on this result, we estimate the minimal number of time steps that are needed to guarantee a particular target error $\epsilon_{target}^{(0)}$ while solving Eq.~\eqref{eq: initial value problem}:
\begin{thm}\label{thm:classicalODEresources}
   Let the assumptions of Theorems~\ref{thm:local_truncation_RKE}-\ref{thm:classicalODEerror} hold. 
   Then, the minimal number of RKM steps $  N_{\tau}^{(0)}$ required to solve Eq.~\eqref{eq: initial value problem} with a target error $\epsilon_{target}^{(0)}$
   is
    \begin{align}\label{1245}
    N_{\tau}^{(0)}=L_{f\tau}T\left(\frac{ K M    \left(e^{b_{max} T L_{fy}s}-1\right)}{\epsilon_{target}^{(0)}b_{max} s  L_{fy}}\right)^{1/p}.
\end{align}
\end{thm}   
Using the latter results, we get:
\begin{corollary}\label{cor:classicalODEcost}
  For a target error $\epsilon_{target}$, the minimal value of the cost function ~\eqref{eq:cost function no shot noise} is
\begin{align}
    C(N_{\tau}^{(0)},s,p)=s L_{f\tau}T \left(\frac{ K M    \left(e^{b_{max}T L_{fy}s}-1\right)}{\epsilon_{target}^{(0)}b_{max} s L_{fy}}\right)^{1/p}\ ,
\end{align}
where $s$ is the number of the RKM stages, $p$ is the RKMs order and  $N_{\tau}^{(0)}$ is the minimal number of time steps of the chosen RKM.
\end{corollary}
The proofs of Theorems~\ref{thm:classicalODEerror} and~\ref{thm:classicalODEresources} are provided in Appendices~\ref{Ax:Analysis with shot noise} and~\ref{app:proof of classical resources}, respectively.
In the following section, we  analyze the RKM in the presence of the   shot noise in the evaluations of the differential $f(\bm{\theta}(\tau))$.

\section{Runge-Kutta methods under the presence of Shot Noise} 

\label{sec: Analysis with Shot Noise Scaling}
In the variational algorithm that we presented in Sec.~\ref{sec: Variational Quantum Simulator preliminaries}, we are solving the ODE given in Eq.~\eqref{eq: parameter initial value problem} by using RKMs. The function $f(\bm{\theta}(\tau))$ from Eq.~\eqref{eq: parameter initial value problem} is given by expectation values estimated via sampling quantum circuits.  Therefore, we assume that instead of $f(\bm{\theta}(\tau))$, we have access to its approximation $\hat{f}(\bm{\theta}(\tau))$.
We need to take into account the statistical errors arising while computing $\hat{f}(\bm{\theta}(\tau))$ based on the measurement results.
The analysis of this section is valid for all differential equations as given in Eq.~\eqref{eq: initial value problem} that have a noisy $\hat{f}(\bm{\theta}(\tau))$.

By virtue of the central limit theorem, let us also assume that each measurement  is drawn from a random Gaussian distribution with mean $f(\bm{\theta}(\tau))$ and standard deviation $\sigma_{\text{single}}$. Calculating the average of these measurements gives the estimate $\hat{f}(\bm{\theta}(\tau))$:
\begin{align}
    \hat{f}(\bm{\theta}(\tau))=\frac{1}{N_r}\sum_{i=1}^{N_r}\hat{f}_i(\bm{\theta}(\tau))\ ,
\end{align}
where $\hat{f}_i(\bm{\theta}(\tau))$ is a single measurement result.
By the central limit theorem as $N_r\rightarrow \infty$, the estimate $\hat{f}(\bm{\theta}(\tau))$ behaves as a Gaussian distribution:
\begin{align}
\sqrt{N_r}(\hat{f}(\bm{\theta}(\tau))-f(\bm{\theta}(\tau)))\rightarrow^d N(0,\sigma_{single}^2).
\end{align} 
From Chebyshev's inequality we get the following bound for a $\delta>0$:
\begin{align}
    P(\|\hat{f}(\bm{\theta}(\tau))-f(\bm{\theta}(\tau))\|\geq \delta)\leq \frac{\sigma^2_{single}}{\delta^2 N_r}. 
\end{align} 
Alternatively, a Chernoff bound would be tighter, but more cumbersome to apply here and probably not lead to a different qualitative analysis.

If we take the probability $P(\|\hat{f}(\bm{\theta}(\tau))-f(\bm{\theta}(\tau))\|\geq \delta)$ to be equal to $\eta\in(0,1)$, then with probability of $1-\eta$, the following bound holds: 
\begin{align}\label{eq: error in evaluating f}
    \|\hat{f}(\bm{\theta}(\tau))-f(\bm{\theta}(\tau))\|\leq\delta=\frac{\Sigma}{\sqrt{N_r}}
\end{align}
In the above, we defined $\Sigma=\frac{\sigma_{single}}{\sqrt{\eta}}$.

Assuming access to a noisy estimate of $f(\bm{\theta}(\tau))$, as described above, we derive error and resource estimates for solving ODEs using RKMs under the presence of shot noise.

In Theorem~\ref{thm:quantumODEerror}, we establish an upper bound on the error $\epsilon_{ODE}^{(\delta)}$ as defined in Eq.~\eqref{eq: initial value problem noisy} by analyzing how both the truncation error of the RKM (see Theorem~\ref{thm:local_truncation_RKE}) and the shot noise in the estimation of $f(\bm{\theta}(\tau))$ (see Eq.~\eqref{eq: error in evaluating f}) contribute to error propagation.

We select the upper bound of the error $\epsilon_{ODE}^{(\delta)}$ as a target error, meaning the maximal error we can expect. 
This target error is the used to determine the minimal number of RKM steps and the minimal number of measurements of each $f(\bm{\theta}(\tau))$ required to keep the resulting error within the target. 
The result is formalized in Theorem~\ref{thm:quantumODEresources}, and the corresponding minimal cost, as defined in Eq.~\eqref{eq:cost function with shot noise}, is derived in Corollary~\ref{cor:quantumODEcost}.

\begin{thm}\label{thm:quantumODEerror}
Let us assume that the assumptions of Theorem~\ref{thm:local_truncation_RKE} hold for Eq.~\eqref{eq: parameter initial value problem}.  Under the conditions  in Eq.~\eqref{eq:boundL}, Eq.~\eqref{eq:definition Lipschitz constant}, the approximation $\hat{\bm{\theta}}_{N_{\tau}}$ at time $\tau_{N_{\tau}}$ calculated via a $p$-th order RKM with $s$ stages in $N_{\tau}$ time steps has the error defined  in Eq.~\eqref{eq: initial value problem noisy} upper bounded by:
\begin{align}
 \epsilon_{ODE}^{(\delta)}&\leq 
\frac{(1+ F(N_{\tau},s))^{N_{\tau}}-1}{F(N_{\tau},s)}\left(\frac{3\delta}{L_{fy}}F(N_{\tau},s)+\left(\frac{T}{N_{\tau}}\right)^{p+1}  K  L^p_{f\tau}   M\right),
\end{align}
where we used the notation introduced in Eq.~\eqref{eq:def of function F}.
\end{thm}
 Let us further denote the latter upper bound as 
    \begin{align}\label{eq:bound on target error for quantum resource thm}
    \epsilon_{target}^{(\delta)}:=  &\frac{(1+ F(N_{\tau},s))^{N_{\tau}}-1}{F(N_{\tau},s)}\left(\frac{3\delta}{L_{fy}}F(N_{\tau},s)+\left(\frac{T}{N_{\tau}}\right)^{p+1}  K  L^p_{f\tau}   M\right)\ .
\end{align}
Given this error estimate, we obtain the following resource minimization and the directly following minimal cost:
\begin{thm}\label{thm:quantumODEresources}
Let us assume that $\hat{f}(\bm{\theta}(\tau))$ is an approximation for $f(\bm{\theta}(\tau))$ calculated from $N_r$ measurements, with the error scaling $\delta=\frac{\Sigma}{\sqrt{N_r}}$, where $\Sigma>0$ is a constant. 
Then, calculating the approximated solution $\hat{\bm{\theta}}_{N_{\tau}}$ to the ODE  in Eq.~\eqref{eq: parameter initial value problem}
by a $p$-th order RKM with $s$ stages and with a target error $\epsilon_{target}^{(\delta)}$  defined in Eq.~\eqref{eq:bound on target error for quantum resource thm}, the number of time steps required must be at least
\begin{align} 
    N_{\tau}^{(\delta)}=T L_{f\tau}\left(\frac{KM\left(e^{b_{max}sL_{fy}T}-1\right)(2p+1)}{\epsilon_{target}^{(\delta)}b_{max}sL_{fy}}\right)^{1/p},
\end{align}
and at least 
\begin{align}
   N_{r}^{(\delta)}=\frac{9\Sigma^2}{L_{fy}^2}   
   \left(\frac{\epsilon_{target}^{(\delta)}}{\left(1+ F(N_{\tau}^{(\delta)},s)\right)^{N_{\tau}^{(\delta)}}-1}-\left(\frac{T}{N_{\tau}^{(\delta)}}\right)^{p+1} \frac{KL^p_{ft}  M}{F(N_{\tau}^{(\delta)},s)}\right)^{-2}\ 
\end{align}
many measurements for the estimation of $f(\bm{\theta}(\tau))$ for each set of inputs $\bm{\theta}(\tau)$. 
\end{thm} 
The proofs of  Theorems~\ref{thm:quantumODEerror} and~\ref{thm:quantumODEresources} are given in Appendices~\ref{Ax:Analysis with shot noise} and~\ref{appendix:proof shot noise resources}, respectively.
\begin{corollary}\label{cor:quantumODEcost}
   For a target error $\epsilon_{target}^{(\delta)}$, the minimal value of the cost function \eqref{eq:cost function with shot noise} is
\begin{align}
     C(N_{\tau}^{(\delta)},N_{r}^{(\delta)},s,p)&= sN_{\tau}^{(\delta)}(s,p)N_{r}^{(\delta)}(s,p),
 \end{align}
 where $s$ is the number of stages, $p$ is the order,  $N_{\tau}^{(\delta)}$ and $N_{r}^{(\delta)}$ are the minimal number of time steps of the RKM; and the minimal number samples in calculating $\hat{f}(\bm{\theta}(\tau))$ for each set of inputs according to Theorem~\ref{thm:quantumODEresources}. 

\end{corollary}
In the following section, we further analyze the latter error and resource estimates. 
By applying it to the variational quantum algorithm discussed in Sec.~\ref{sec: Variational Quantum Simulator preliminaries}, we show
by numerical estimations in Sec.~\ref{sec: numerics}, how they evaluate for specific DEs such as given in Eq.~\eqref{eq: parameter initial value problem}. 
We shall see that with these results, we can choose a Runge-Kutta method such that the required total number of circuit evaluations will be minimal.

\section{Analysis of the Variational Quantum Algorithm}\label{sec:Introduce Quantum Circuits, McLachlan Principle, Shot noise, Lipschitz Number}

In this section, we analyze the total error as defined in Eq.~\eqref{eq:total error definition} of the algorithm provided in Sec.~\ref{sec: Variational Quantum Simulator preliminaries}. We show its relation to Eq.~\eqref{eq: initial value problem noisy}  and derive the resource requirements of the algorithm for a given target error, based on the estimates in Sec.~\ref{sec: Analysis with Shot Noise Scaling}. Afterwards, we continue with an estimation of the shot noise and further quantities based on a toy model we introduce. These further estimates are necessary for the numerical analysis of the following Sec.~\ref{sec: numerics}.
 
\subsection{Error and Resource Estimate}\label{sec: Variational Quantum Simulator}
    
We are interested in estimating the total error in Eq.~\eqref{eq:total error definition}.
Using the triangle inequality, we get
    \begin{align}\label{1353}
        \epsilon_{total} \leq&\|{\psi}(T)-{\phi}({\bm{\theta}}(T))\|_{1}+\|{\phi}({\bm{\theta}}(T))-{\phi}(\hat{\bm{\theta}}_{N_{\tau}})\|_{1}= \epsilon_{PQC}+\epsilon_{par}      \ .
    \end{align}
    Here, we define the distance between the trial functions with the  different input parameters  as follows:
    \begin{align}
       \epsilon_{par}:= \|{\phi}(\bm{\theta}(T))-{\phi}(\hat{\bm{\theta}}_{N_{\tau}})\|_1.
    \end{align}
    The representation error coming from approximating the function ${\psi}(T)$ with the trial function ${\phi}({\bm{\theta}}(T))$ is defined as
    \begin{align}
       \epsilon_{PQC}:=   \|{\psi}(T)-{\phi}({\bm{\theta}}(T))\|_1.
    \end{align}
 Furthermore, using the multi-dimensional mean-value theorem
    \footnote{Multivariate Mean Value Theorem: for $x,y\in R^n$ 
    \begin{align}
        \|f(x)-f(y)\|_q\leq \sup\limits_{z\in[x,y]}\|f'(z)\|_{(q,p)}\|x-y\|_p,
    \end{align}
    Where $z\in[x,y]$  denotes a vector $z$
     contained in the set of points between $x,y\in R^n$ and $\|f'(z)\|_{(q,p)}$
     is the $L_{(p,q)}$  norm of the derivative matrix of $f$ evaluated at $z$. }, 
     the second term of Eq.~\eqref{1353} can be bounded by
    \begin{align}
       \epsilon_{par}&
       \leq 
       \sup_{\bm{\theta_0}\in \Xi}\|\nabla_{\bm{\theta}}{\phi}(\bm{\theta_0}(T))\|_{(1,2)}\epsilon_{ODE}^{(\delta)},
    \end{align}
    where
    $\Xi=\{w\hat{\bm{\theta}}_{N_{\tau}}+(1-w){\bm{\theta}}(T)\vert  w\in [0,1]\}$. 
    The norm of the Jacobian is defined as
    \begin{align}\label{eq:jacobian definition}
        \sup_{\bm{\theta_0}\in \Xi}\|\nabla_{\bm{\theta}}{\phi}(\bm{\theta_0}(T))\|_{(1,2)}:=&\sup_{\bm{\theta}_0\in \Xi}\sup_{\bm{\theta}^*(T)\in\Xi,\bm{\theta}^*(T)\neq 0}\frac{\|\nabla_{\bm{\theta}}{\phi}(\bm{\theta_0}(T))\cdot\bm{\theta}^*(T)\|_1}{\|\bm{\theta}^*(T)\|_2}\ .
    \end{align}
The following lemma shows us an upper bound to this expression.
\begin{lemma}
\label{lemma: jacobian bound}
        The norm of the Jacobian with the chosen circuit (Eq.~\eqref{eq:circuit on total state} and  Eq.~\eqref{eq:circuit on one qubit}) is bounded by
        \begin{align}
            \sup_{\bm{\theta}_0\in \Xi}\sup_{\bm{\theta}^*(T)\in\Xi,\bm{\theta}^*(T)\neq 0}\frac{\|\nabla_{\bm{\theta}}{\phi}(\bm{\theta_0}(T))\cdot\bm{\theta}^*(T)\|_1}{\|\bm{\theta}^*(T)\|_2}\leq \sup_{\bm{\theta}^*(T)\in\Xi, \bm{\theta}^*(T)\neq 0}\frac{\sum_{k}\left(\sum_{j}2|f_{k,j}|\right)|\theta^*_{k}(T)|}{\|\bm{\theta}^*(T)\|_2}\ .
        \end{align}
    \end{lemma}
We provide the proof in Appendix~\ref{appendix: error bound theorems}. 

According to this lemma, Eq.~\eqref{eq:jacobian definition} is upper bounded by
\begin{align}
    S:= \left(\sup_{\bm{\theta}^*(T)\in\Xi,\bm{\theta}^*(T)\neq 0}\frac{\sum\limits_{k=1}^{N_V}\left(\sum\limits_{j=1}^{N_d}2|f_{k,j}|\right)|\theta^*_{k}(T)|}{\|\bm{\theta}^*(T)\|_1}\right)\ .
\end{align}
Thus, the total error is bounded by
\begin{align}\label{eq:total_error_easy_form}
         \epsilon_{total}&\leq \epsilon_{PQC}+S  \epsilon_{ODE}^{(\delta)} .
    \end{align}
As described in Sec.~\ref{subsec: Total error and resources}, there are errors which we assume to be negligible. In particular, we are disregarding the representation error $\epsilon_{PQC}$. We are thus left with estimating $S$ for the specific Ansatz at hand (see Sec.~\ref{sec: numerics}) and the error in calculating the parameters, $\epsilon_{ODE}^{(\delta)}$, which we analyzed in Sec.~\ref{sec: Analysis with Shot Noise Scaling}.

The resource that we would like to minimize for the application of the variational algorithm is the total number of circuit evaluations $N_{circ}$ which is needed for running the algorithm. In comparison to the cost that was estimated in Sec.~\ref{sec: Analysis with Shot Noise Scaling}, we have to add an additional factor in order to get $N_{circ}$. The derivation of this factor follows.

Denote with $N_A$ and $N_C$ the numbers of circuits that are needed to calculate the matrix $A$ and the vector $C$ for specific input parameters. Each matrix $A$ or vector $C$ consists of $N_V^2$ or $N_V$ elements, respectively where $N_V$ is the number of parameters $\bm{\theta}$. Each element of $A$ and $C$ is calculated by the evaluation of $N_d^2$ or $N_dN$ different circuits, respectively. Here $N$ is the number of terms in the Hamiltonian and $N_d$ is an upper bound on the number of Pauli strings in the unitary operator $R_k$ of the variational circuits, as defined in Eq.~\eqref{eq:circuit on one qubit}. Thus, they are bounded above by $N_A\leq\overline{N_A} =N_V^2N_d^2$ and $N_C\leq\overline{N_C} = N_VN_dN$.
    Let us assume that these upper bounds are  always reached. We denote by $N_{\tau}^{(\delta)}$ the minimum number of time steps in the ODE solving algorithm and by $N_{r}^{(\delta)}$ the minimum number of times, each stage value $f\left(\bm{\theta}(\tau)\right)$ is measured; both have been estimated in Theorem~\ref{thm:quantumODEresources}. 
    By $s$, we denote the number of stages that are calculated at each time step of the RKM. 
    Therefore, the number of circuit evaluation $N_{circ}$ needed to calculate $\Vec{\theta}(T)$ is equal to
    \begin{align}\label{eq: cost function with nvfactor}
        N_{circ}&= N_{\tau}^{(\delta)}s N_{r}^{(\delta)}(\overline{N_A}+\overline{N_C})\\
        &= N_{\tau}^{(\delta)}s N_{r}^{(\delta)}(N_V^2N_d^2+N_VN_dN)\\
           &= N_{\tau}^{(\delta)}s N_{r}^{(\delta)}N_VN_d(N_VN_d+N)\ .
    \end{align}
One can see that the number of circuit evaluations in the whole circuit scales quadratically in the number of circuit parameters $N_V$.
    
Note that, if instead of counting the total number of circuit evaluations we care about the total run time of the algorithm, and allow parallel evaluations of $N_V^2N_d^2+N_VN_dN$ circuits, the cost comes down to $ N_{\tau}^{(\delta)}s N_{r}^{(\delta)}$.

In the next subsection, we will further estimate the dependence of $N_{r}^{(\delta)}$ on the condition number of $A$, and give a comprehensive overview of the resource estimate in Eqs.~\eqref{eq:comprehensiveresourcecountfirstequation}-~\eqref{eq:comprehensiveresourcecountlastequation}.

\subsection{Estimation of the Shot Noise}
\label{subsec: Estimation of the Shot noise}
In Sec.~\ref{sec: Analysis with Shot Noise Scaling}, we presented error and resource estimates for RKMs with an error in the evaluations of $f(\tau,\bm{\theta(\tau)})$ stemming from shot noise. We derive  our estimates on the basis of the  bound in Eq.~\eqref{eq: error in evaluating f}.  Since we would like to incorporate the results of Sec.~\ref{sec: Analysis with Shot Noise Scaling} into analyzes of  the required  resources for the variational quantum algorithm, we are studying the shot noise inherent to this algorithm in this section. In the end of this section, we will give an estimate of $\Sigma$ as defined in Eq.~\eqref{eq: error in evaluating f}.
    
Each term $\left(f^*_{k,i}f_{l,j}\bra{\overline{0}}R^{\dagger}_{k,i}R_{l,j}\ket{\overline{0}}+h.c.\right)$ or $
\left(f^*_{k,i}\lambda_m\bra{\overline{0}}R_{k,i}^{\dagger}\sigma_mR\ket{\overline{0}}+h.c.\right)$  in the matrix $A$ and the vector $C$ as given in Eqs.~\eqref{eq: A evaluation} and~\eqref{eq: C evaluation} can be written in the following form:
    \begin{align}
    \label{eq: general form A and C evaluation}
        q=    a \mathrm{Re}\left(e^{i\zeta}\bra{0}^{\otimes n} U\ket{0}^{\otimes n}\right),
    \end{align}
    where the amplitude $a$ and the phase $\zeta$ are determined by either $f^*_{k,i}f_{l,j}$ or $f^*_{k,i}\lambda_m$. The unitary operator $U$ is equal to either $R^{\dagger}_{k,i}R_{l,j}$ or $R_{k,i}^{\dagger}\sigma_mR$.  
    
The terms $q$ can be obtained by the parameterized quantum circuits shown in  Fig.~\ref{fig:circuit A}. Note that these circuits are different from the circuit which prepares the final state $\ket{\phi(\hat{\bm{\theta}}_{N_{\tau}})}$.
An ancillary qubit of the circuit is initialized as $(\ket{0}+e^{i\zeta}\ket{1})/\sqrt{2}$, while the remaining qubits are initialized in the state $\ket{0}^{\otimes n}$. A projective measurement of the ancillary qubit in the $\{\ket{+},\ket{-}\}$ basis has following probability  to find the qubit in the state $\ket{+}$:
    \begin{align}
      \hat{P}_{+}=\frac{\mathrm{Re} \big(e^{i\zeta}\bra{0}^{\otimes n}U\ket{0}^{\otimes n}\big)+1}{2}.
    \end{align} 
To estimate the probability $\hat{P}_{+}$ and subsequently the terms $q$, we need to implement and measure the circuit $N_r$ times.
Let us assume that the measurements correspond to independent Bernoulli distributed random variables, which have the variance
    $    \sigma^2(\hat{P}_{+})=\hat{P}_{+}(1-\hat{P}_{+})/N_r$.
    
Thus, the variances of the estimates of the elements of Eq.~\eqref{eq: A evaluation} are scaling as
    \begin{align}
      \sigma^2\left(\mathrm{Re} \left(e^{i\zeta}\bra{0}^{\otimes n}U\ket{0}^{\otimes n}\right)\right) = 4 \sigma^2(\hat{P}_{+})=\mathcal{O}(\frac{1}{N_r}).
    \end{align}
    After estimating $\mathrm{Re} \big(e^{i\zeta}\bra{0}^{\otimes n}U\ket{0}^{\otimes n}\big)$, the amplitudes of $f^*_{k,i}f_{l,j}$ or $f^*_{k,i}\lambda_m$ have to be calculated classically and multiplied. 
    Since they are not depending on each other, all of the circuits with the same input parameters $\bm{\theta}(\tau)$ can be run in parallel. 
    \par We denote the realization of the matrices $A$ and $C$ in  Eq.~\eqref{2144} as $\hat{A}$ and $\hat{C}$, respectively.

    The variances of the estimates $\hat{A}_{k,l}(\bm{\theta}(\tau))$ and $\hat{C}_k(\bm{\theta}(\tau))$ are scaling as
    \begin{align}
        \sigma^2\left(\hat{A}_{k,l}(\tau)\right)\mathcal{O}\left(\frac{N_d^2}{N_r}\right),\quad   \sigma^2\left(\hat{C}_{k}(\tau)\right)\mathcal{O}\left(\frac{N_d N}{N_r}\right).
    \end{align}
    Then by Chebyshev's inequality, we can conclude that the probability of estimation error is bounded by
    \begin{align}
        &P\left(|A_{k,l}-\hat{A}_{k,l}|\geq b\right)\leq 
        \frac{N_d^2}{N_rb^2},\quad P\left(|C_{k}-\hat{C}_{k}|\geq b\right)\leq 
        \frac{N_d N}{N_r b^2}.
    \end{align}
    We can conclude that we face up with estimation errors  arise due to 
    shot noise and the representation errors from calculating the matrix elements on the quantum circuits. 
The following lemma holds:
    \begin{lemma}
    \label{lemma: shot noise_2}
    With the probability of at least $1-\eta$, $\eta\in[0,1]$, the following bounds hold:
    \begin{align}
          \|A-\hat{A}\|&\leq\frac{\|\{\sigma_{k,l}\}_{k,l=1}^{N_V}\|}{\sqrt{N_r  \eta}},\\\nonumber
     \|C-\hat{C}\|&\leq \frac{\|\{\sigma_{k}\}_{k=1}^{N_V}\|}{\sqrt{N_r  \eta}},
    \end{align}
    where we used the notations
    \begin{align}\label{14450}
            \sigma_{k,l}=\sqrt{\sum_{i,j=1}^{N_d}|f_{k,i}^*f_{l,j}|^2},\quad   \sigma_{k}=\sqrt{\sum_{i=1}^{N_d}\sum_{m=1}^{N}|f_{i,k}^*\lambda_{m}|^2}.
        \end{align}
    \end{lemma}
    For a detailed proof, see Appendix \ref{appendix:Shot noise estimates}.

This lemma provides us with upper bounds on the errors on $A$ and $C$ that arise from shot noise.
These errors translate into an error in $f$ as defined in Eq.~\eqref{eq: parameter initial value problem}, which can be estimated with the following lemma:

\begin{lemma}\label{lemma:liubovs sensitivity}
Consider the linear equation
\begin{align}
    Af=C\ ,
\end{align}
    where $A$ is a a non-singular $N_V\times N_V$-dimensional matrix and $C$ an $N_V$-dimensional vector.
    Let us introduce disturbances in $A$ and $C$ by $A\mapsto A+\xi R$ and $C\mapsto C+\xi r$,  leading to the disturbed linear equation 
    \begin{align}
    \label{eq:disturbanceinthm}
    (A+\xi R)\hat{f}&=C+\xi r\ ,
    \end{align}
where $\hat{f}$ is the solution vector under the disturbance,
$\xi\geq 0$ is a real scalar, $R$ is an $N_V\times N_V$ dimensional matrix and $r$ is an $N_V$ dimensional vector.
Then, we can write the following estimate:
\begin{align}
    \frac{\|\hat{f}-f\|}{\|f\|}\leq \xi\kappa(A)\left(\frac{\|r\|}{\|C\|}+\frac{\| R\|}{\|A\|}\right)+\mathcal{O}(\xi^2)\ ,
\end{align}
where $\kappa(A):=\|A\|\|A^{-1}\|$ is the condition number of $A$.
\end{lemma}
We give a proof of this lemma in Appendix \ref{appendix: error bound theorems}.

If $A$ is invertible, we can identify $f=A^{-1} C$, and can apply this Theorem to get:
\begin{align} \label{eq:upperboundfromliubovsensitivity}
    \|f-\hat{f}\|\leq \|A^{-1}  C\|\xi\kappa(A)\left(\frac{\|r\|}{\|C\|}+\frac{\| R\|}{\|A\|}\right)+\mathcal{O}(\xi^2)\ .
\end{align}
In order to apply the Theorem to our analysis above, we are defining $R=\{\sigma_{k,l}\}_{k,l=1}^{N_V}$, $r=\{\sigma_{k}\}_{k=1}^{N_V}$ and 
\begin{align}
    \xi=\frac{1}{\sqrt{N_r\eta}}\ .
\end{align}

If in a specific instance, $A$ is singular, it is not possible to invert it and thus apply Lemma~\ref{lemma:liubovs sensitivity}. However, one can  apply a regularization procedure by deviating the matrix elements of $A$ slightly in such a way that it becomes non-singular. 
This will be leading to an additional error which, as we described in Sec.~\ref{subsec: Total error and resources}, is one of the error sources we are not taking into consideration in our analysis. For a way of dealing with this issue, see techniques used in Ref.~\cite{alghassi2022variational,mcardle2019variational,fontanela2021quantum, anuar2024operatorprojectedvariationalquantumimaginary}.

Taken together, we arrive at the following number of total circuit evaluations, combining Eqs.~\eqref{eq:def of function F},~\eqref{eq: error in evaluating f},~\eqref{eq: cost function with nvfactor},~\eqref{eq:upperboundfromliubovsensitivity}, and the results of Theorem~\ref{thm:quantumODEresources}:
\begin{align}\label{eq:comprehensiveresourcecountfirstequation}
N_{circ}
           &= N_{\tau}^{(\delta)}s N_{r}^{(\delta)}N_VN_d(N_VN_d+N)\\
     N_{r}^{(\delta)}&=\frac{9\Sigma^2}{L_{fy}^2}   
   \left(\frac{\epsilon_{target}^{(\delta)}}{\left(1+ F(N_{\tau}^{(\delta)},s)\right)^{N_{\tau}^{(\delta)}}-1}-\left(\frac{T}{N_{\tau}^{(\delta)}}\right)^{p+1} \frac{KL^p_{ft}  M}{F(N_{\tau}^{(\delta)},s)}\right)^{-2}\\
   \Sigma &= \|A^{-1}  C\|\frac{\kappa(A)}{\sqrt{\eta}}\left(\frac{\|r\|}{\|C\|}+\frac{\| R\|}{\|A\|}\right)+\mathcal{O}(\xi^2)\\
    F(N_{\tau},s) &= 
        \frac{b_{max}}{a_{max}} \left(\left(1+L_{fy}a_{max} \frac{T}{N_{\tau}}\right)^{s}-1\right) .\label{eq:comprehensiveresourcecountlastequation}
\end{align}

To provide further details and to make the analysis concrete, in Section~\ref{sec:estimations on toy model}, we select a specific toy model as a representative example. 
We estimate the quantities $\|A\|$, $\|C\|$, $\|A^{-1}C\|$, $\kappa(A)$ and $L_{fy}$ based on this toy model and in Sec.~\ref{sec: numerics}, we will estimate $\|R\|$, $\|r\|$, $a_{max}$, $b_{max}$, $K$, $M$ and $L_{f\tau}$ for general Runge-Kutta methods and a chosen Ansatz.

\subsection{Estimations based on a toy model}
\label{sec:estimations on toy model}
In order to apply the error and resource estimates from Sec.~\ref{sec: Analysis with Shot Noise Scaling} in practice (such as we do in Sec.\ref{sec: numerics}), we need to make estimates on several quantities. Especially challenging is the estimation of the condition number $\kappa(A)$, $\|A\|$, $\|C\|$ and $\|A^{-1}  C\|$ that are needed for the estimates on shot noise as derived in Sec.~\ref{subsec: Estimation of the Shot noise} and the estimation of the Lipschitz constant $L_{fy}$, as defined in Eq.~\eqref{eq:definition Lipschitz constant}.

As those quantities depend on $A$ and $C$, we are taking the following considerations:
Both $A$ and $C$ depend on expectation values from quantum circuits as written in Eq.~\eqref{eq: general form A and C evaluation}. The parameters $\bm{\theta}(\tau_k)$ at specific time steps $\tau_k$ vary a lot depending on the instance being solved. Since the expectation values that constitute the coefficients of $A$ and $C$ are depending on these parameters as well as the chosen Ansatz, it is non-viable to estimate them in general.

However, motivated by the fact that expectation values of parametrized quantum circuits defined as in our chosen Ansatz (see Eq.~\eqref{eq:Ansatz}) resemble truncated Fourier series (see Refs.~\cite{schuld2021effect,landman2022classically}), we simulate plausible coefficients of $A$ and $C$ according to the following toy model:
 
Let us assume that each of the elements $A_{k,l}$ and $C_{k}$ take the form of functions $w_{k,l}(\tilde{\theta})$ and $w_{k}(\tilde{\theta})$ in one input parameter $\tilde{\theta}$, respectively:
\begin{align}\label{eq:toy model A}
    w_{k,l}(\tilde{\theta}):=\alpha_{1,k,l}\cos(\alpha_{2,k,l}\tilde{\theta}+\alpha_{3,k,l})+\alpha_{4,k,l}\sin(\alpha_{5,k,l}\tilde{\theta}),\\\label{eq:toy model C}
    w_{k}(\tilde{\theta}):=\alpha_{1,k}\cos(\alpha_{2,k}\tilde{\theta}+\alpha_{3,k})+\alpha_{4,k}\sin(\alpha_{5,k}\tilde{\theta})\ .
\end{align}
We are choosing the amplitudes $\alpha_{1,k,l}$, $\alpha_{1,k}$ and $\alpha_{4,k,l}$, $\alpha_{4,k}$ to be randomly drawn from the Gaussian distribution $\mathcal{N}(1,0.1)$, and the phases $\alpha_{2,k,l}$, $\alpha_{2,k}$ and $\alpha_{5,k,l}$, $\alpha_{5,k}$ and phase differences $\alpha_{3,k,l}$, $\alpha_{3,k}$ to be randomly drawn from the Gaussian distribution $\mathcal{N}(0,0.1)$. 

\subsubsection{Condition number and norms on \texorpdfstring{$A$}{A} and \texorpdfstring{$C$}{C}}
In the literature, it is often assumed that the condition number of a matrix scales polynomially in its dimension~\cite{vu2007condition}.
If this assumption holds for the  $N_V\times N_V$-dimensional matrix $A$, there exists a constant $\Gamma$, such that with high probability, the following bound holds:
\begin{align}\label{eq:boundonconditionnumber}
    \kappa(A)\leq N_V^\Gamma\ .
\end{align}
In order to find out if this bound also holds if the coefficients of $A$ are sampled according to the toy model as described above, we plot such $\kappa(A)$ and the functions $N_V^2,N_V^3$ and $N_V^4$ for varying $N_V$ in Fig.~\ref{fig:condition number}. 
\begin{figure}[htbp]
    \centering
    \begin{overpic}[scale=0.7,grid=false,tics=10]{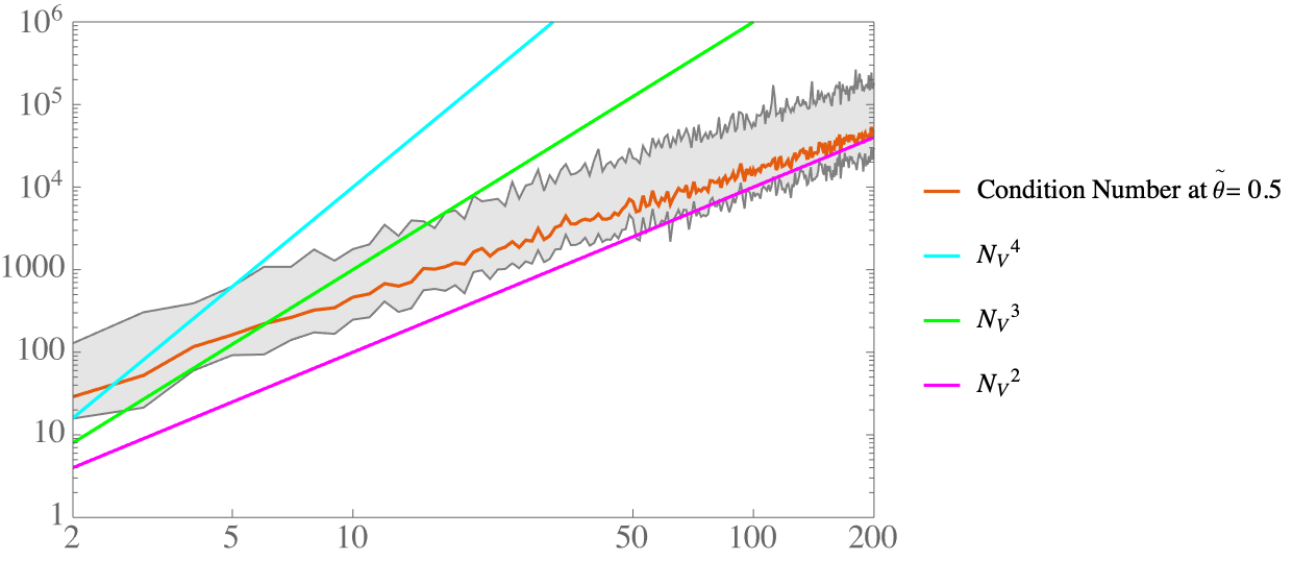}
        \put(37, -1){\makebox(0,0){Number of parameters $N_V$}} % x-axis label
        \put(-1, 25){\rotatebox{90}{\makebox(0,0){\text{Condition number of $A$}}}} % y-axis label
    \end{overpic}
    \caption{A log plot comparing the condition number obtained from the toy model with $\tilde{\theta}=1/2$ and different upper bounds due to $N_V^2$, $N_V^3$ and $N_V^4$. We sampled 100 different matrices $A$ according to the toy model and calculated the resulting condition number for each sample. The orange line shows the median and the gray shaded area shows the range between the 0.16 and 0.84 quantiles of the condition number estimates. We used the Frobenius norm.}
    \label{fig:condition number}
\end{figure}

Most of the condition numbers lie above the threshold $N_V^2$, while for $N_V\geq 10$, most condition numbers lie below the threshold $N_V^3$. 
Typically, more than $10$ parameters are chosen in applications of the variational algorithm (We are using $N_V=25$ in our numerical analysis in Sec.~\ref{sec: numerics}). In rare cases, $A$ constructed by the toy model is ill-conditioned, and $\kappa(A)$ significantly exceeds those bounds. As described earlier, those cases can be mitigated by a regularization of $A$ with the drawback of introducing an additional error.

We therefore assume that the bound in Eq.~\eqref{eq:boundonconditionnumber} holds with high probability for our toy model, with the constant $\Gamma=3$.

Analogously, we estimate the norms $\|A\|$, $\|C\|$ and $\|A^{-1}  C\|$ based on $A$ and $C$ that follow the introduced toy model. 
In  Fig.~\ref{fig:norm of A and C and AinverseC}, we plot those norms and the functions $f_1(N_V)=N_V$ and $f_2(N_V)=\sqrt{N_V}$ for varying $N_V$. 

\begin{figure}[htbp]
    \centering
    \begin{overpic}[scale=0.65,grid=false,tics=10]{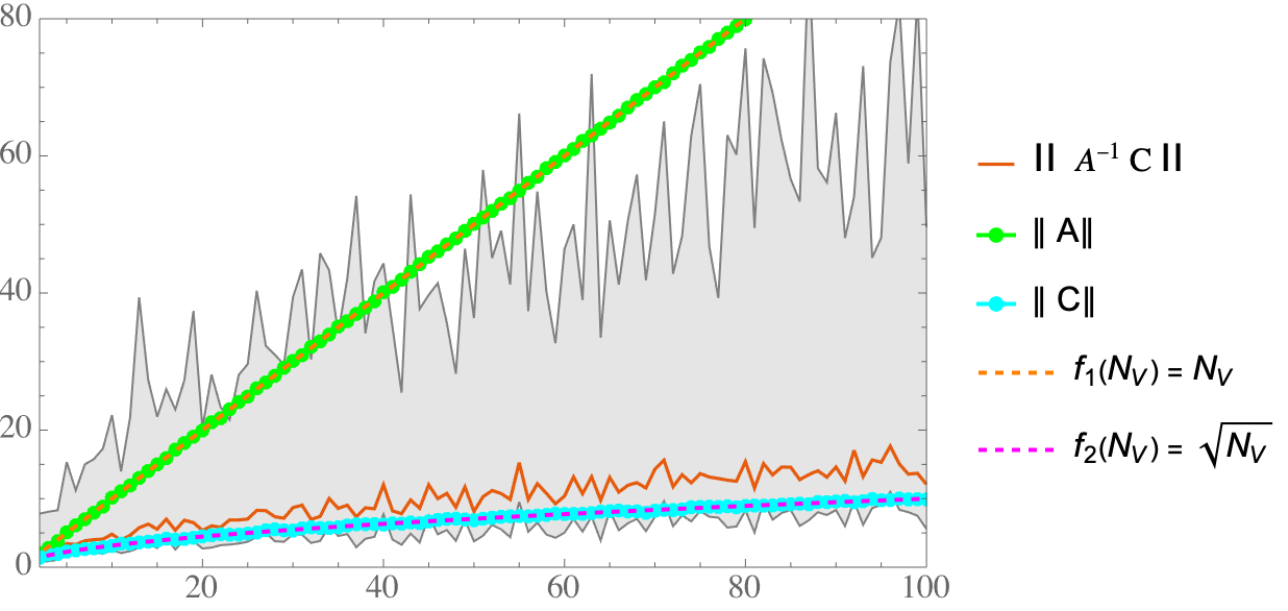}
        \put(37, -1){\makebox(0,0){Number of parameters $N_V$}} % x-axis label
        \put(-2.5, 24){\rotatebox{90}{\makebox(0,0){\text{Values of Norms and approximating functions}}}} % y-axis label
    \end{overpic}
    \vspace{0.1cm}
    \caption{A plot comparing the norms $\|A\|$, $\|C\|$ and $\|A^{-1}  C\|$ obtained from the toy model with the functions $f_1(N_V)=N_V$ and $f_2(N_V)=\sqrt{N_V}$. In order to estimate $\|A^{-1}  C\|$, we did 100 samples for $A$ and $C$ and calculated the resulting norm for each sample. The orange line shows the median and the gray shaded area shows the range between the 0.16 and 0.84 quantiles of the estimates for $\|A^{-1}  C\|$. For $\|A\|$, we used the Frobenius norm and for $\|C\|$ and $\|A^{-1}  C\|$ the $2$ norm.}
    \label{fig:norm of A and C and AinverseC}
\end{figure}

We see that it is reasonable to assume the scaling $\|A\|=\Theta( N_V)$ and $\|C\|=\Theta(\sqrt{N_V})$. 
For $\|A^{-1}  C\|$, we cannot make similarly reasonable estimates due to the high fluctuation. 
However, we can assume an upper bound $\|A^{-1}  C\|\leq 60$, which holds with high probability in the number of parameter range $N_V\in [0,100]$. Similar to managing the condition number, it is also possible to enforce this bound by a regularization of the matrix $A$.

The above analysis on estimates and bounds for $\kappa(A)$, $\|A\|$, $\|C\|$ and $\|A^{-1}  C\|$ included into the estimate in Eq.~\eqref{eq: error in evaluating f} leads to the following inequality to hold with high probability:
\begin{align}\label{eq:upper bound f error}
    \|f-\hat{f}\|<\delta\leq \frac{60}{\sqrt{N_r\eta}}N_V^3\left(\frac{\|\{\sigma_{k}\}_{k=1}^{N_V}\|}{\sqrt{N_V}}+\frac{\|\{\sigma_{k,l}\}_{k,l=1}^{N_V}\|}{N_V}\right)\ .
\end{align}
Hence, we conclude the following upper bound for $\Sigma$ (see Eq.~\eqref{eq: error in evaluating f}) with high probability:
\begin{align}
\label{eq:formula for Sigma}
    \Sigma \leq \frac{60}{\sqrt{\eta}}N_V^3\left(\frac{\|\{\sigma_{k}\}_{k=1}^{N_V}\|}{\sqrt{N_V}}+\frac{\|\{\sigma_{k,l}\}_{k,l=1}^{N_V}\|}{N_V}\right).
\end{align}

\subsubsection{Estimation of the Lipschitz Constant}
\label{subsec:Estimation of the Lipschitz constant}
We have seen in Sec.~\ref{sec: Analysis with Shot Noise Scaling}, that the Lipschitz constant $L_{fy}$ of $f(\bm{\theta}(\tau))$ with respect to the $\bm{\theta}(\tau)$ has a direct influence on the error and resource estimates. In fact, as we will see in Sec.~\ref{sec: numerics}, the sensitivity of the total cost is higher with $L_{fy}$ than with most of the other parameters. Let us therefore analyze the Lipschitz constant for the variational algorithm described in Sec.~\ref{sec: Variational Quantum Simulator preliminaries}. 
In this case, $f(\bm{\theta}(\tau))$ is defined as
\begin{align}
        f(\bm{\theta}(\tau))=f\left(\bm{\theta}(\tau)\right):= A^{-1}\left(\bm{\theta}(\tau)\right)C\left(\bm{\theta}(\tau)\right)\ ,
\end{align}
and the Lipschitz constant $L_{fy}$ turns into the upper bound
\begin{align}
L_{fy}\geq\frac{\left\|f\left(\bm{\theta}_1(\tau)\right)-f\left(\bm{\theta}_2(\tau)\right)\right\|}{\|\bm{\theta}_1(\tau)-\bm{\theta}_2(\tau)\|}\ , \quad \forall \bm{\theta}_1(\tau), \bm{\theta}_2(\tau)\ .
\end{align}
We are again modeling the matrix $A$ and the vector $C$ with the toy model described in Eqs.~\eqref{eq:toy model A} and~\eqref{eq:toy model C}. This toy model has a single input variable $\tilde{\theta}$ instead of an $N_V$-dimensional vector $\bm{\theta}(\tau)$. Let us vary this parameter, while we keep the other (randomly drawn) parameters of the toy model fixed, and define the following function:
\begin{align}
\label{eq:lipfunction}
\text{Lip}\left(\tilde{\theta}_1,\tilde{\theta}_2\right):=\frac{\left\|f\left(\tilde{\theta}_1\right)-f\left(\tilde{\theta}_2\right)\right\|}{\|\tilde{\theta}_1-\tilde{\theta}_2\|}\ .
\end{align}
We show two 3D plots of this function, with varying $\tilde{\theta}_1$ and $\tilde{\theta}_2$ and fixed dimensions of $A$ and $C$ with $N_V=25$ in  Fig.~\ref{fig:lipschitzplot}. For each of the two plots, the random toy model parameters are drawn independently. 
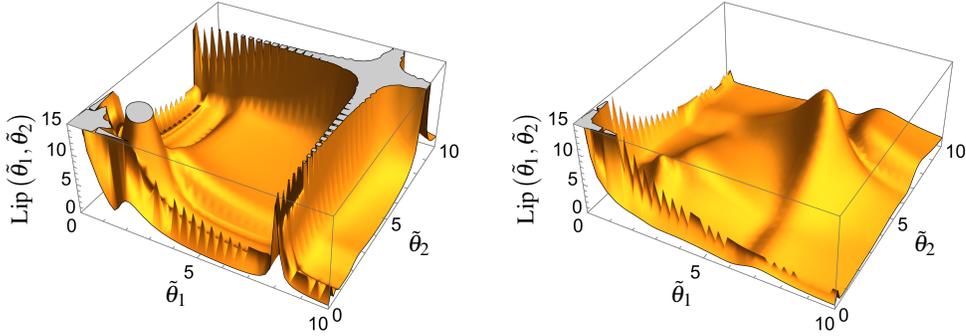
\begin{figure}[htbp]
    \centering
    \begin{subfigure}[b]{0.4\textwidth}
        \centering
        \begin{overpic}[width=\textwidth,grid=false,tics=10]{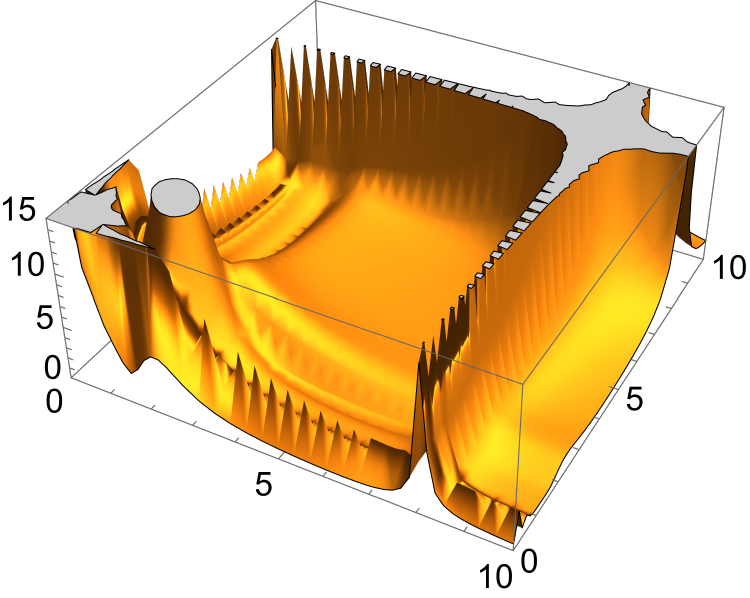}
            \put(32, 10){\makebox(0,0){$\tilde{\theta}_1$}} % x-axis label
            \put(-5, 40){\rotatebox{90}{\makebox(0,0){$\text{Lip}\left(\tilde{\theta}_1,\tilde{\theta}_2\right)$}}} % y-axis label
            \put(90, 22){\makebox(0,0){$\tilde{\theta}_2$}} % top label
        \end{overpic}
    \end{subfigure}
    \hspace{1cm} % Adjust horizontal space between subfigures
    \begin{subfigure}[b]{0.4\textwidth}
        \centering
        \begin{overpic}[width=\textwidth,grid=false,tics=10]{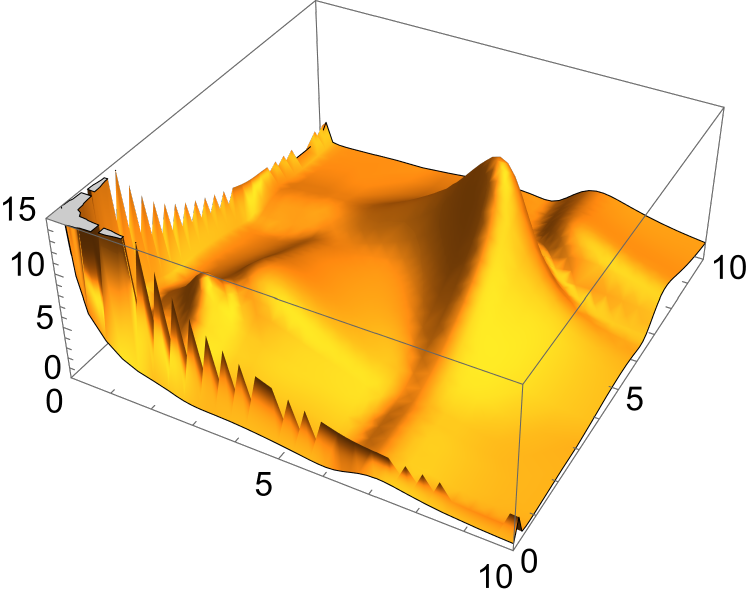}
            \put(32, 10){\makebox(0,0){$\tilde{\theta}_1$}} % x-axis label
            \put(-5, 40){\rotatebox{90}{\makebox(0,0){$\text{Lip}\left(\tilde{\theta}_1,\tilde{\theta}_2\right)$}}} % y-axis label
            \put(90, 22){\makebox(0,0){$\tilde{\theta}_2$}} % top label
        \end{overpic}
    \end{subfigure}
    \caption{The function $\text{Lip}\left(\tilde{\theta}_1,\tilde{\theta}_2\right)$ defined in Eq.~\eqref{eq:lipfunction} plotted for two different drawings of random matrix and vector parameters according to the toy model from Eqs.~\eqref{eq:toy model A} and~\eqref{eq:toy model C}. We fixed $N_V=25$ and varied the parameters $\tilde{\theta}_1$ and $\tilde{\theta}_2$ between $0$ and $10$.}
    \label{fig:lipschitzplot}
\end{figure}

At the gray areas, the function $\text{Lip}\left(\tilde{\theta}_1,\tilde{\theta}_2\right)$ exceeds the value of $15$. We see that in order to be able to bound $\text{Lip}\left(\tilde{\theta}_1,\tilde{\theta}_2\right)$ from above by a Lipschitz constant $L_{fy}$, one has to guarantee that at each update step from $\tilde{\theta}_1$ to $\tilde{\theta}_2$, the function $\text{Lip}\left(\tilde{\theta}_1,\tilde{\theta}_2\right)$ stays inside the region below the threshold $L_{fy}=15$. It is possible to take higher estimates for $L_{fy}$ and create possibly larger areas for valid $\tilde{\theta}_1$ and $\tilde{\theta}_2$. However, this comes to the account of less tight error and resource bounds in Theorems \ref{thm:quantumODEerror} and \ref{thm:quantumODEresources}. 

Interestingly enough, the function $\text{Lip}\left(\tilde{\theta}_1,\tilde{\theta}_2\right)$ keeps similar patterns for changing the mean and the variance of the distributions from which the amplitudes of the toy model are drawn. 
However, it shows an increasingly more rugged landscapes for both increasing the mean and the variance of the random distributions of the phases of the toy model. 
As expected, the function $\text{Lip}\left(\tilde{\theta}_1,\tilde{\theta}_2\right)$ behaves periodically with respect to the phase difference parameter of the toy model, with the period depending on the phases.

Under the assumption that the toy model we chose models the behavior of $A$ and $C$ sufficiently well, we use the estimates for $\kappa(A)$, $\|A\|$, $\|C\|$, $\|A^{-1}  C\|$ and $L_{fy}$ derived in this section in order to concretize the error and resource estimates from Sec.~\ref{sec: Analysis with Shot Noise Scaling}, and to numerically analyze them in the following section.

\section{Numerical Analysis of the Error and Resource Estimates}
\label{sec: numerics}
In this section, we are numerically analyzing the error and resource estimates from Sec.~\ref{sec: analysis classical ode solver} and Sec.~\ref{sec: Analysis with Shot Noise Scaling} for an application of the RKMs to solving a simple ODE and for an application of the variational algorithm presented in Sec.~\ref{sec: Variational Quantum Simulator preliminaries} to solving a partial DE coming from finance.
\subsection{Solving a Classical ODE without Shot Noise}
\label{subsection:vanilla ode numerics}
Let us analyze the cost function provided in Corollary~\ref{cor:classicalODEcost} with a simple ODE: 
 \begin{align}
        \frac{\partial \bm{\theta}(\tau)}{\partial \tau}=\frac{\bm{\theta}(\tau)}{2}, \quad
        \bm{\theta}(0)=1\ .
\end{align}
This ODE has the exponential function $\bm{\theta}(\tau)=\exp(\frac{\tau}{2})$ as a solution. We choose this simple ODE, because it is easy to analyze and to estimate the corresponding parameters that are used for the cost function.
\subsubsection{Parameter Estimation and Sensitivity Analysis for the Classical ODE Solver}
\label{subsubsection: Parameter estimation and sensitivity classical}
The Lipschitz constant $L_{fy}$ is equal to $0.5$. If we pick a final time $T=5$, then the function $f(\bm{\theta}(\tau))=\frac{1}{2}\bm{\theta}(\tau)$ is upper bounded by $M=13$. The derivative of $f(\bm{\theta}(\tau))$ with respect to $\tau$ is equal to $0.25\exp(\frac{1}{2}\tau)$ (since we know $\bm{\theta}(\tau)=\exp(\frac{1}{2}\tau)$). Thus we can choose a Lipschitz constant $L_{f\tau}=3.1$, which upper bounds this derivative for times $\tau\in[0,T=5]$.

It can be easily checked that the other bounds required in Theorem~\ref{thm:local_truncation_RKE} hold as well.

The approximation of $b_{max}:=\max_{i}|b_i|$ is depending on the chosen Runge-Kutta method. While there exist Runge-Kutta methods that have coefficients $|b_i|>1$, those are the exception as one can see by looking at several methods (for example, in~\cite{butcher2016numerical}).

To estimate the constant $K$, let us look at its definition~\onlinecite[Chapter 318]{butcher2016numerical}:
        \begin{align}
        K:=\sum_{|t|=p+1}\frac{1}{\sigma(t)}\Big|\Phi(t)-\frac{1}{t!}\Big| ,
    \end{align}
    where $t$ is a rooted tree of order $|t|$,
    the quantities $\sigma(t)$ and $t!$ are defined as products of factorials, and lie between $1$ and $|t|!$.
    The quantity $\Phi(t)$ depends on the coefficients of the chosen method. 
    For a rigorous definition of these terms, see ~\cite{butcher2016numerical}. 
    As it is discussed in~\onlinecite[Chapter 318]{butcher2016numerical}, it is not possible to create a general rule for estimating the terms in $K$. 
    However, what can be said is that the number of terms in the sum, the number of unlabeled rooted trees with $p+1$ nodes, is asymptotically equal to~\cite{on-line_encyclopedia} $0.439\times 2.956^{p+1}\times (p+1)^{-3/2}$ .
    This series scales exponentially in $p$. Thus, if the summands $\frac{1}{\sigma(t)}\Big|\Phi(t)-\frac{1}{t!}\Big|$ were constant in $p$, $K$ would scale exponentially in $p$ as well. 
    However, in ~\cite{butcher2016numerical} it is said that it is reasonable to assume that $K$ is constant in $\Delta \tau$. In Ref.~\onlinecite[Chapter 244]{butcher2016numerical}, several Runge-Kutta (and related) methods were developed for which the corresponding error constants were estimated, all of which are upper-bounded in magnitude by $1$.
    In the following analysis, we therefore assume that $K\leq 5$ holds with high probability. But, as we will see later, the total cost does not change as much in the parameter $K$ as it does in other parameters, and we found that the implications of our work still hold qualitatively with a pessimistic upper bound of $K\leq 0.439\times 2.956^{p+1}\times (p+1)^{-3/2}$. 
    Let us further choose a target error of $\epsilon_{target}=0.001$.
    
    We collect the estimates for the analysis in this section in Table~\ref{tab:estimates classical}.
    \begin{table}[htbp]
    \centering
    \begin{tabular}{|c|c|}
        \hline
        \textbf{Parameter} & \textbf{Estimate} \\
        \hline
        \( b_{max}\) & \(   1 \) \\
        \hline
        \( L_{fy} \) & \(   0.5 \) \\
        \hline
        \( T \) & \(   5 \) \\
        \hline
        \( K \) & \(   5 \) \\
        \hline
        \( L_{f\tau} \) & \(   3.1 \) \\
        \hline
        \( M \) & \(   13 \) \\
        \hline
         \( \epsilon_{target} \) & \(   0.001 \)\\
        \hline
    \end{tabular}
    \caption{Estimates of various parameters that optimize the savings by using a higher-order Runge-Kutta method instead of the Euler method. Used for the analysis in  Fig.~\ref{fig:costdifferentp classical}.}
    \label{tab:estimates classical}
\end{table}
Based on these default values, we are analyzing the sensitivity of the total cost on tuning the parameters in  Fig.~\ref{fig:parametersensitivity classical}. We show how the total cost changes when a single default value is multiplied with a scaling factor, whereas the remaining parameters are kept at default. We see that the cost does go down for higher-order Runge-Kutta methods and that the parameter with the most effect is the total time $T$, the Lipschitz constant $L_{fy}$, the parameter $b_{max}$ stemming from the Runge-Kutta method and the total target error $\epsilon_{target}$. The spike in the cost at order $p=6$ is due to the non-linear relationship between the order of the method and the number of stages (see Table~\ref{tab:Relation between order and number of stages of Runge-Kutta methods}).
\begin{figure}[htbp]
    \centering
    \begin{overpic}[scale=0.65,grid=false,tics=10]{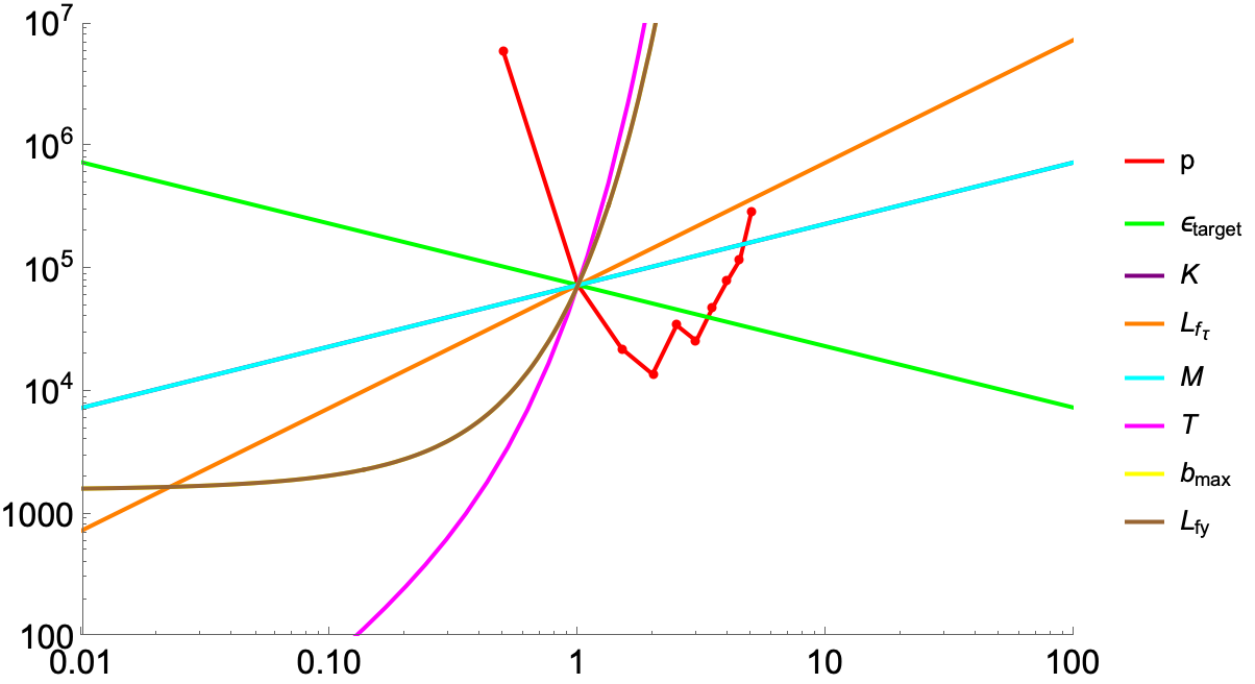}
        \put(50, -1){\makebox(0,0){Scaling factor $x$}} % x-axis label
        \put(-5, 33){\rotatebox{90}{\makebox(0,0){\text{Value of the cost function from Corollary}\ref{cor:classicalODEcost}}}} % y-axis label
    \end{overpic}
    \caption{A plot comparing the sensitivity of the cost function with respect to different parameters. The intersecting point in the middle is the value of the cost function where all parameters are chosen default as given in Table~\ref{tab:estimates classical} as well as $p=2$ as a default. In each graph, we are changing one parameter, while keeping all the other parameters at default. We are changing the parameter by multiplying with the scaling factor $x$ given as the abscissa. The graphs for $M$ and $K$ overlap as well as the graphs for $L_{fy}$ and $b_{max}$. The continuous red line for $p$ is just for visualization purposes, as $p$ is integer.}
    \label{fig:parametersensitivity classical}
\end{figure}

\subsubsection{Numerical Analysis of the Classical ODE Solver}
Using the parameters estimated above in Table~\ref{tab:estimates classical}, we estimate the total cost for Runge-Kutta methods of different orders. For the relation between the minimum number of stages $s$ needed for a particular order $p$, we use Table~\ref{tab:Relation between order and number of stages of Runge-Kutta methods}. We show the results in  Fig.~\ref{fig:costdifferentp classical}. On the left hand side is a table with the cost depending on the order and the ratio of the Euler method compared to a Runge-Kutta method of order $p$. On the right hand side, we plotted this ratio depending on the Runge-Kutta order. We see that for our particular example, all methods of order $2\leq p\leq 10$ are more cost efficient than the Euler method, where a Runge-Kutta method of order $4$ is the most cost efficient with savings of a factor of ca. $10^{3}$ compared to the cost of the Euler method. There is a second spike for the order $p=6$, because of the non-linear relationship between the order of the method and the number of stages as shown in Table~\ref{tab:Relation between order and number of stages of Runge-Kutta methods}.
\begin{figure}[htbp]
    \centering
    %\raggedleft

    \begin{subfigure}[t]{0.45\textwidth}
        \centering
         %\vspace{-6.3cm} % Adjust vertical alignment
        %\hspace{-15.3cm} % Adjust horizontal alignment
        
\begin{tabular}{|c|c|c|c|}
    \hline
    $p$ & \text{cost(p)} & \text{cost(1)/cost(p)} & $N_{\tau}^{(0)}$ \\
    \hline
    1  & $2.25 \times 10^7$ & $1.00$            & $2.25 \times 10^7$ \\
    \hline
    2  & $9.60 \times 10^4$ & $2.35 \times 10^2$ & $4.80 \times 10^4$ \\
    \hline
    3  & $1.99 \times 10^4$ & $1.13 \times 10^3$ & $6.63 \times 10^3$ \\
    \hline
    4  & $1.01 \times 10^4$ & $2.22 \times 10^3$ & $2.54 \times 10^3$ \\
    \hline
    5  & $1.38 \times 10^4$ & $1.64 \times 10^3$ & $2.29 \times 10^3$ \\
    \hline
    6  & $1.03 \times 10^4$ & $2.18 \times 10^3$ & $1.47 \times 10^3$ \\
    \hline
    7  & $1.36 \times 10^4$ & $1.65 \times 10^3$ & $1.52 \times 10^3$ \\
    \hline
    8  & $1.71 \times 10^4$ & $1.32 \times 10^3$ & $1.56 \times 10^3$ \\
    \hline
    9  & $2.07 \times 10^4$ & $1.09 \times 10^3$ & $1.60 \times 10^3$ \\
    \hline
    10 & $3.33 \times 10^4$ & $6.76 \times 10^2$ & $2.08 \times 10^3$ \\
    \hline
\end{tabular}

    \end{subfigure}
    \vspace{0.5cm}
    \begin{subfigure}[t]{0.9\textwidth}
        \centering
     %   \vspace{-1.3cm} % Adjust vertical alignment
      %  \hspace{-1.3cm} % Adjust horizontal alignment
        \begin{overpic}[scale=0.6,grid=false,tics=10]{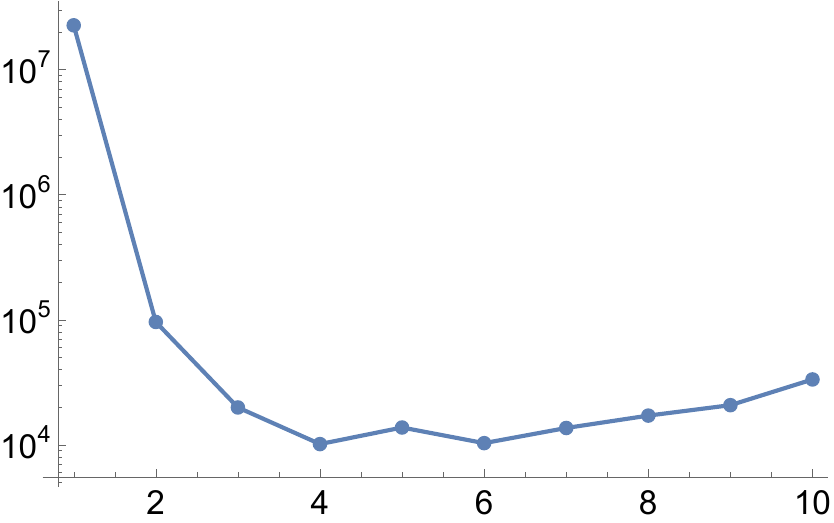}
            \put(50, -5){\makebox(0,0){\text{Runge-Kutta method order} p}} % x-axis label
            \put(-5,33){\rotatebox{90}{\makebox(0,0){\text{cost}(p)}}} % y-axis label
        \end{overpic}
    \end{subfigure}
    \vspace{0.1cm}
    \caption{Comparison of the cost for different RKM orders, where the estimated parameters are fine tuned to maximize the savings by using a higher order as given in Table~\ref{tab:estimates classical}. In the first column, we have the Runge-Kutta method order $p$. In the second column is the total cost of an algorithm that uses a Runge-Kutta method of order $p$, in the third column the ratio of an algorithm that uses the Euler method $(p=1)$ with an algorithm that uses a Runge-Kutta method of order $p$, and in the third column is the minimal number of time steps $ N_{\tau}^{(0)}$. In the plot, we show $N_{circ}(p)$ plotted against the order of the Runge-Kutta method. The cost is calculated according to Corollary~\ref{cor:classicalODEcost}.}
    \label{fig:costdifferentp classical}
\end{figure}

\subsection{Solving a Linear PDE with the Variational Quantum Algorithm}
\label{subsec:option pricing}
We are numerically analyzing the estimates from Sec.~\ref{sec: Analysis with Shot Noise Scaling} and Sec.~\ref{sec:Introduce Quantum Circuits, McLachlan Principle, Shot noise, Lipschitz Number} for an application of the variational algorithm presented in Sec.~\ref{sec: Variational Quantum Simulator preliminaries} to solving the Black Scholes model, which is a linear partial DE coming from finance. The motivation for choosing this model is twofold: On the one hand, we are using this example in order to specify the estimates done in previous sections and in showing the consequences of the choice of RKMs on the total cost. On the other hand, we are using this example to stress the wide range of DEs on which the variational algorithm presented in Sec.~\ref{sec: Variational Quantum Simulator preliminaries} can be applied to.

\subsubsection{Black Scholes Model and Error Analysis}
Let us begin with describing the problem of option pricing and the application of the variational algorithm to it. This application has been studied before in several ways ~\cite{fontanela2021quantum, radha2021quantum, alghassi2022variational, 10112619}.

A European call option is used in practice in the following way: Initially, at time $t=0$, an option is acquired, specifying a particular underlying asset, an expiration time $t_{final}$, and a strike price $K$. When time reaches $t=t_{final}$, the option buyer faces a decision: whether to exercise the option by buying the asset at the strike price $K$ or to refrain from exercising it. The buyer's rational choice at time $t_{final}$ is determined by the price $S(t_{final})$ of the underlying asset at that moment. If $S(t_{final}) > K$, indicating that the final asset price exceeds the strike price, the buyer can use the option to purchase the asset for the price $K$ and immediately sell it at the higher market price of $S(t_{final})$ in an ideal market. This transaction results in a profit of $S(t_{final}) - K$ for the buyer. Conversely, if $S(t_{final}) \leq K$, the buyer would not choose to exercise the option by purchasing the asset since they would not be able to generate a profit from selling it in the market.

One can conclude that the payoff of the buyer is equal to
    \begin{align}
    \label{eq: payoff call option}
        V(t_{final},S)= \max \{S(t_{final})-K, 0\}\ .
    \end{align}
     Since the stock price at time $t_{final}$ is unknown, one models the stock price stochastically.
    A simple model is the so-called Black-Scholes model which characterizes the stock price $S(t)$ as a stochastic variable that follows a geometric Brownian motion:
    \begin{align}
        dS(t)=\mu S(t)dt+\sigma S(t)dW_{t} \ ,
    \end{align}
    where $\mu$ is the drift of the stock price, $\sigma$ is its standard deviation (called 'volatility')  and the random variable $dW_t$ is a Wiener process.
    The arbitrage assumption states that it is impossible to build a portfolio which gives positive return without risk. By incorporating this assumption and using It\^{o} calculus, the stochastic DE in $S(t)$ is transformed to a parabolic partial DE (PDE) that models the price of a call option $V(t,S)$:
    \begin{align}\label{1209}
      \frac{\partial V}{\partial t} + \frac{1}{2}\sigma^2S^2\frac{\partial^2V}{\partial S^2} + rS\frac{\partial V}{\partial S}= rV\ .
    \end{align}
    This PDE is called the Black Scholes equation.
    The parameters of this model are the volatility $\sigma$ of the stock price and the risk-free interest rate $r$. Both are assumed to be independent of time. The Black Scholes can be mapped to the imaginary time Schrödinger equation with the following transformations:
    
    Applying the transformations  $\tau=(t_{final}-t)\sigma^2$, $x=\log(S)$, and subsequently $u(\tau, x)=e^{-ax-b\tau}V(\tau, x)$ with the parameters $a=\frac{1}{2}-\frac{r}{\sigma^2}$ and  $b=-\frac{1}{2}a^2-\frac{r}{\sigma^2}$, where $u(\tau, x)$ is a modified price, one obtains ~\cite{fontanela2021quantum}:
    \begin{align}
    \label{eq: transformed BS equation}
        \frac{\partial }{\partial \tau}u(\tau, x)=\frac{1}{2}\frac{\partial^2 }{\partial x^2}u(\tau, x)\ .
    \end{align}
    Note that $\tau\in[0,T]$, where we write $T=t_{final}\sigma^2$.
    This equation is equivalent to the imaginary time Schr\"{o}dinger equation in  Eq.~\eqref{eq: generalized imaginary time schroedinger equation, ket formulation}, where $\ket{\psi(\tau)}$ is a quantum state representing the option price and $\mathcal{H}=-\frac{1}{2}\frac{\partial^2 }{\partial x^2}$ is the Hamiltonian operator. 
    Solving  Eq.~\eqref{eq: generalized imaginary time schroedinger equation, ket formulation} is hence equivalent to solving the Black-Scholes equation   in Eq.~\eqref{1209} after reversing the transformations.
   
    We can interpret the state with the option price in the following way:
    
    Consider a register of $n$ qubits. On this chain, we define a set of $2^n$ pairwise orthogonal states $\{\ket{x}\}_x$. We select an interval of possible stock prices and discretize the interval into $2^n$ points. Consequently, we associate the basis states with stock prices, where the state $\ket{0}^{\otimes n}$ corresponds to the minimum stock price, and the state $\ket{1}^n$ corresponds to the maximum stock price of the chosen interval.
    
    The quantum state from Eq.~\eqref{1240} encodes the option value corresponding to a particular stock price in the amplitudes of $\ket{\psi(\tau)}$ in the following way:
     \begin{align}
        \ket{\psi(\tau)}=\sum_x\sqrt{p_x(\tau)}\ket{x}\ ,\quad  \sum_xp_x(\tau)=1,\quad \forall\tau.\label{eq:quantum_option_state}
    \end{align}
    The boundary condition is $p_x(0)=\gamma(0)V(0,x)e^{-ax}$,
    where $\gamma(0)$ is the normalization constant and $V(0,x)$ is the option price at time $\tau=0$. Thus, by obtaining the probability $p_x(\tau)$ at time $T=t_{final}\sigma^2$ for a particular stock price $x=\log(S)$, the corresponding option price can be calculated as:
    \begin{align}
        V(t_{final},x)=\gamma^{-1}(T)  p_x(T)  e^{ax+bt_{final}\sigma}\ .
    \end{align}
 In the following subsections, we are estimating the parameters based on the application of the variational algorithm to this option pricing use case and do a numerical analysis.

\subsubsection{Parameter Estimations and Sensitivity Analysis for the Variational Quantum Algorithm}
\label{subsubsec: Parameter estimates and Numerical Analysis quantum}
We are now making educated guesses for the parameters that are related to the use case of option pricing as described above and analyze the sensitivity of our error and resource analysis from Sec.~\ref{sec: Analysis with Shot Noise Scaling} with respect to these parameters.

In this and the following section, we use the symbol $\sim$ for approximations of a term with a scalar, including rounding errors.

In several papers~\cite{fontanela2021quantum, radha2021quantum, alghassi2022variational, 10112619}, the application of the variational algorithm as described has been applied to option pricing. We are basing our parameter estimates on Ref.~\onlinecite[Sec. 5.1.]{fontanela2021quantum}, where the authors take an Ansatz with the parameters
    $ N=16$, $ N_V=25$ and $N_d=1$.
    That means that the matrix $A$ is a $25\times 25$ matrix, the vectors $C$ and $\bm{\theta}$ have $25$ elements, as well as the matrix $\sigma_{k,l}$ and the vector $\sigma_k$ that we defined in Lemma~\ref{lemma: shot noise_2}. 
     We take the probability $\eta$ from the same lemma equal to  
     $\eta\sim 0.05$.
     
    A typical total time $t_{final}$ in option pricing is one year and a typical volatility is $\sigma=0.2\times 1/\text{year}^2$. Thus, $t_{final}\times\sigma^2=T \sim 0.04$.
    
In Ref.~\onlinecite[Chapter 3]{fontanela2021quantum}, the parameters $f_{k,j}$ are defined as $f=i/2$. With this, let us approximate the perturbations $\|\sigma_{k,l}\|$ and $\|\sigma_k\|$:

it makes sense to see them in terms of the maximal deviation of individual entries. Given $f=i/2$, they are for the matrix equal to $N_d^21/2\times2$ and for the vector $N_dN1/2\times2$. Thus, we assume the total norms to scale as $\|\sigma_{k,l}\|\leq N_VN_d^2\times 1$ and $\|\sigma_{k}\|\leq N_VN_dN\times 1$. 

For approximations of the parameters $a_{max}$, $b_{max}$ and $K$, see our discussion in Sec.~\ref{subsection:vanilla ode numerics}. We pick now the parameters $a_{max}=b_{max}=1$ and the $K=5$. But as above, we see later that $N_{circ}$ does not depend strongly on $K$ and our results are still qualitatively valid for a large variety of $K$.

In line with the discussion in Sec.~\ref{sec: Variational Quantum Simulator}, we estimate the Lipschitz constant $L_{fy}$ to be equal to $L_{fy}=15$, the condition number to be upper bounded as $\kappa(A)\leq N_V^3=15625$ and $M=60$  with the caveat of having to assure the toy model is chosen adequate enough and the bounds can be guaranteed by regularization to satisfy these bound.

We are also estimating a probability coming from the shot noise of $\eta=0.05$. Combining the estimates, we are therefore getting with high probability the upper bound
\begin{align}
    \label{eq: approximation for Sigma}
    \Sigma \leq \frac{60}{\sqrt{\eta}}N_V^3\left(\frac{N_VN_dN}{\sqrt{N_V}}+\frac{N_VN_d^2}{N_V}\right)\sim 3.4\times 10^8\ .
\end{align}

The parameter $L_{f\tau}$ is defined as an upper bound to the time derivative of $f\left(\bm{\theta}(\tau)\right)$. Since the time dependence of $f$ is fully carried via the function $\bm{\theta}(\tau)$, we can apply a reasoning similar to that for estimating $L_{fy}$. We assume that the same behavior and the same bounds hold and, therefore, estimate $L_{f\tau}=L_{fy}=15$.

We cannot determine if the other bounds hold that are required in Theorem~\ref{thm:local_truncation_RKE}, since we do not have full knowledge of the function $f\left(\bm{\theta}(\tau)\right)$.

For estimating the quantity $S$ as defined in Eq.~\eqref{eq:total_error_easy_form}, we need to take the following into account:
We estimate that for all $1\leq k\leq N_V$, $|\theta^*_{k}(T)|\sim \frac{\|\bm{\theta}\|_2}{N_V}$. Together with taking $f_{k,j}=i/2$ as above, we get $S\sim 1$.

Lastly, we are choosing a target error of $\epsilon_{target}\sim 0.001$ .
Let us thus make the following estimates, collected in Table~\ref{tab:estimates}:
\begin{table}[htbp]
    \centering
    \begin{tabular}{|c|c|}
        \hline
        \textbf{Parameter} & \textbf{Estimate} \\
        \hline
        \( a_{max}\) & \(   1 \) \\
        \hline
        \( b_{max}\) & \(   1 \) \\
        \hline
        \( L_{fy} \) & \(   15 \) \\
        \hline
        \( T \) & \(   0.04 \) \\
        \hline
        \( K \) & \(   5 \) \\
        \hline
        \( L_{f\tau} \) & \(   15 \) \\
        \hline
        \( M \) & \(   60 \) \\
        \hline
        \( S \) & \(   1 \) \\
        \hline
        \( \Sigma \) & \(   3.4\times 10^8 \)\\
        \hline
         \( \epsilon_{target} \) & \(   0.001 \)\\
        \hline
    \end{tabular}
    \caption{Estimates of various parameters based on the option pricing application. Used for the analysis in Figs.~\ref{fig:parametersensitivity} and \ref{fig:costdifferentp}.}
    \label{tab:estimates}
\end{table}

Let us now show a plot where we examine the scaling of $N_{circ}$ with respect to the input parameters. It is difficult to make general statements about the scaling of $N_{circ}$ from Eq.~\eqref{eq: cost function with nvfactor}, since it does not have a trivial form, thus a look at a plot helps with understanding its scaling.
We take as default the parameters estimated above and analyze in  Fig.~\ref{fig:parametersensitivity} the deviations caused by changing a single parameter away from the default.

\begin{figure}[htbp]
    \centering
    \begin{overpic}[scale=0.65,grid=false,tics=10]{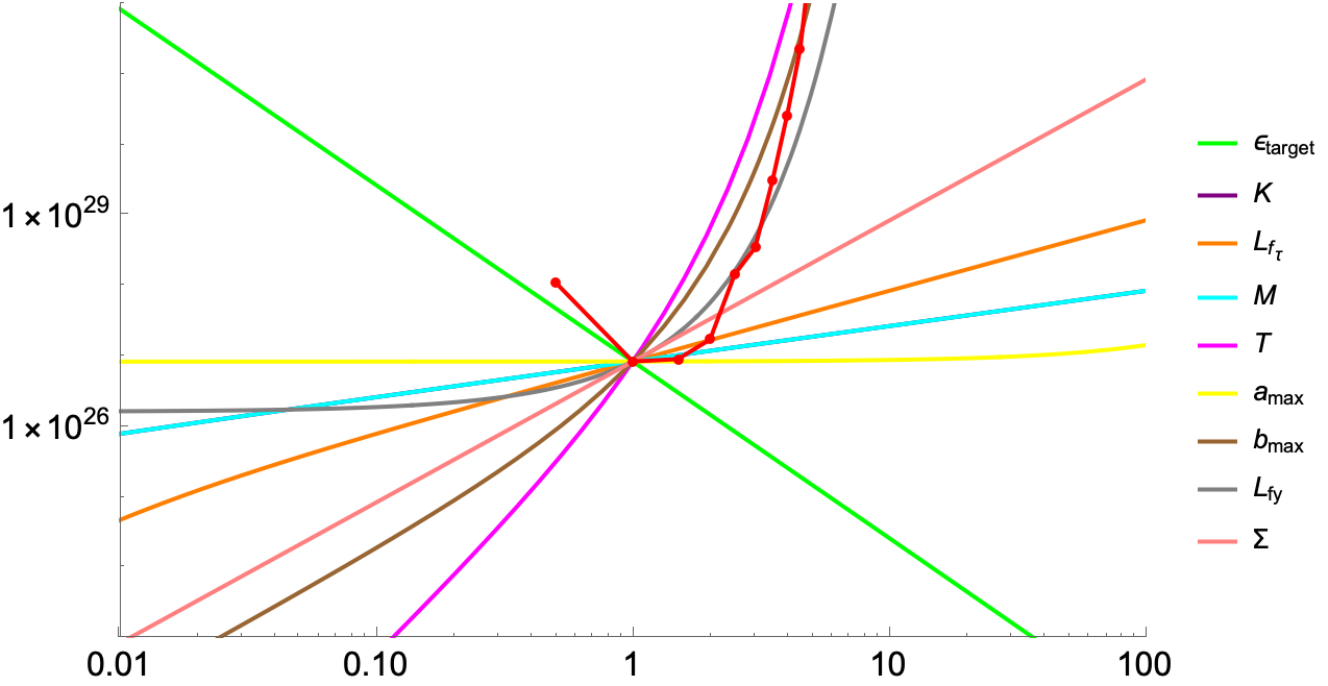}
        \put(50, -1){\makebox(0,0){Scaling factor $x$}} % x-axis label
        \put(-5, 33){\rotatebox{90}{\makebox(0,0){\text{Value of $N_{circ}$ from Eq.} \eqref{eq: cost function with nvfactor}}}} % y-axis label
    \end{overpic}
    \caption{A plot comparing the sensitivity of $N_{circ}$ with respect to different parameters. The intersecting point in the middle is the value of $N_{circ}$ where all parameters are chosen default as given in Table~\ref{tab:estimates} as well as $p=2$ as a default. In each graph, we are changing one parameter, while keeping all the other parameters at default. We are changing the parameter by multiplying with the scaling factor $x$ given as the abscissa. The following colours are corresponding to $N_{circ}$ with one parameter changed: $p$, $\epsilon_{target}$, $L_{f\tau}$, $M$ and $K$, $T$, $a_{max}$, $b_{max}$, $L_{fy}$, $\Sigma$, where $x$ is the scaling factor at the abscissa. The graphs for $M$ and $K$ overlap. The continuous red line for $p$ is just for visualization purposes, as $p$ is integer.}
    \label{fig:parametersensitivity}
\end{figure}
We conclude from this figure, that $N_{circ}$ does not change much for the parameters $a_{max}$, $M$, $L_{f\tau}$, $K$ and $\Sigma$. As it was proved impossible to estimate $K$ in a general way, it is promising to see that our estimate of $N_{circ}$ does not scale with it as much as with other parameters.

In addition, $N_{circ}$ changes faster for the parameters $p$, $T$, $b_{max}$, $\epsilon_{target}$, and $L_{fy}$. For $p$ in particular, we see that a minimum of $N_{circ}$ is reached for a value higher than for $p=1$. We will examine the behavior with $p$ with more detail in the next section.

\subsubsection{Numerical Analysis of the Variational Quantum 
Algorithm}
\par For the numerical analysis, we use the parameters given in Table~\ref{tab:estimates} that are fit to the option pricing use case in order to calculate $N_{circ}$, which is the total number of circuit evaluations needed for the algorithm. As we mentioned in Sec.~\ref{subsubsec: Parameter estimates and Numerical Analysis quantum}, the parameter $K$ is the most challenging to estimate. For this analysis, we choose $K = 5$. A comparison of the costs of algorithms based on Runge-Kutta and Euler methods is given in Fig.~\ref{fig:costdifferentp}. We see that the highest saving in $N_{circ}$  compared to the Euler method can be done with a Runge-Kutta method of order $p=2$.

To reach a target error of $\epsilon_{target}=0.001$, we need an $ N_{r}^{(\delta)}\sim 10^{22}$, which is equivalent to requiring classical machine precision for the entries of matrices $A$ and vectors $C$ ($\epsilon\sim 1/\sqrt{ N_{r}^{(\delta)}}\sim 5\times 10^{-12}$). 
The total number of circuit evaluations $N_{circ}$ is equal to $1.62\times 10^{28}$. 
\par In Ref.~\onlinecite[Supplementary information VIA]{arute2019quantum}, the time that the superconducting Sycamore chip can be used before having to be recalibrated is around $1$ day. 
That implies that all calculations have to be done in at most $24$ hours. 
The time for one readout of the Sycamore chip takes around $4\mu s$~\cite{sycamoredatasheet}. 
That means that it is possible to evaluate around $2\times 10^{10}$ quantum circuits. 

It is obvious that the high accuracy is needed because of the inversion of the matrix A and the resulting error propagation. The number $N_{r}^{(\delta)}$ given in Theorem~\ref{thm:quantumODEresources} depends quadratically on the factor $\Sigma$ estimated in $\eqref{eq: approximation for Sigma}$ to be upper bounded with high probability by $ 3.4\times 10^8$. If it was possible to reach a bound $\Sigma\leq 1$, the number $N_{r}^{(\delta)}$ would decrease to the order $10^7$ which would be feasible for quantum hardware, . 

In order to illustrate the potential resource savings that can be gained by choosing a higher-order RKM instead of the Euler method, we provide a second analysis based on a different choice of parameters provided in Table~\ref{tab:estimates maximize gain}.  
This choice is  inspired by the estimates in Table~\ref{tab:estimates} changing some of the parameters within reasonable ranges in order to increase the resource savings. 
The resulting resource requirements are shown in  Fig.~\ref{fig:costdifferentp maximize gain}. 
We conclude that by choosing a Runge-Kutta method of order $p=4$, there have to be done a factor of $\sim 2.56 \times 10^{3}$ less circuit evaluations than when choosing the Euler method. 
However, the number of shots for each circuit ($N_{r}^{(\delta)}\sim 1.98\times 10^{26}$) is still too high in order to be realized on any quantum hardware.

\begin{figure}[htbp]
    \centering
    \vspace{0cm} % Adjust vertical alignment
    \hspace{0cm} % Adjust horizontal alignment
    \begin{subfigure}[t]{0.9\textwidth} % Full width for the table
        \centering
        \begin{tabular}{|c|c|c|c|c|c|}
    \hline
    $p$ & $N_{circ}(p)$ & $N_{circ}(1)/N_{circ}(p)$ & $N_{r}^{(\delta)}$ & $N_{\tau}^{(\delta)}$ & \text{Circuits} \\
    \hline
    1 & $2.13 \times 10^{29}$ & $1  $ & $7.03 \times 10^{21}$ & $2.96 \times 10^{4}$ & $3.03 \times 10^{7}$ \\
    \hline
    2 & $1.62 \times 10^{28}$ & $13.18 $ & $3.87 \times 10^{22}$ & $2.04 \times 10^{2}$ & $4.19 \times 10^{5}$ \\
    \hline
    3 & $1.75 \times 10^{28}$ & $12.21 $ & $1.53 \times 10^{23}$ & $37.06 $ & $1.14 \times 10^{5}$ \\
    \hline
    4 & $3.31 \times 10^{28}$ & $6.45  $ & $5.19 \times 10^{23}$ & $15.55 $ & $6.38 \times 10^{4}$ \\
    \hline
    5  & $3.38 \times 10^{29}$  & $6.31\times 10^{-1}$             & $5.48 \times 10^{24}$ & $10.03$         & $6.17 \times 10^4$ \\
    \hline
    6  & $7.79 \times 10^{29}$  & $2.74\times 10^{-1}$             & $1.56 \times 10^{25}$ & $6.96$          & $4.99 \times 10^4$ \\
    \hline
    7  & $7.49 \times 10^{30}$  & $2.85\times 10^{-2}$             & $1.41 \times 10^{26}$ & $5.74$          & $5.30 \times 10^4$ \\
    \hline
    8  & $7.00 \times 10^{31}$  & $3.05\times 10^{-3}$             & $1.25 \times 10^{27}$ & $4.98$          & $5.62 \times 10^4$ \\
    \hline
    9  & $6.45 \times 10^{32}$  & $3.31\times 10^{-4}$             & $1.08 \times 10^{28}$ & $4.47$          & $5.96 \times 10^4$ \\
    \hline
    10 & $2.16 \times 10^{34}$  & $9.9 \times 10^{-6}$ & $3.03 \times 10^{29}$ & $4.33$          & $7.11 \times 10^4$ \\
    \hline
\end{tabular}
    \end{subfigure}

    \vspace{0.5cm} % Add some vertical space between table and plot
    \begin{subfigure}[t]{0.9\textwidth} 
        \centering
        \begin{overpic}[scale=0.6,grid=false,tics=10]{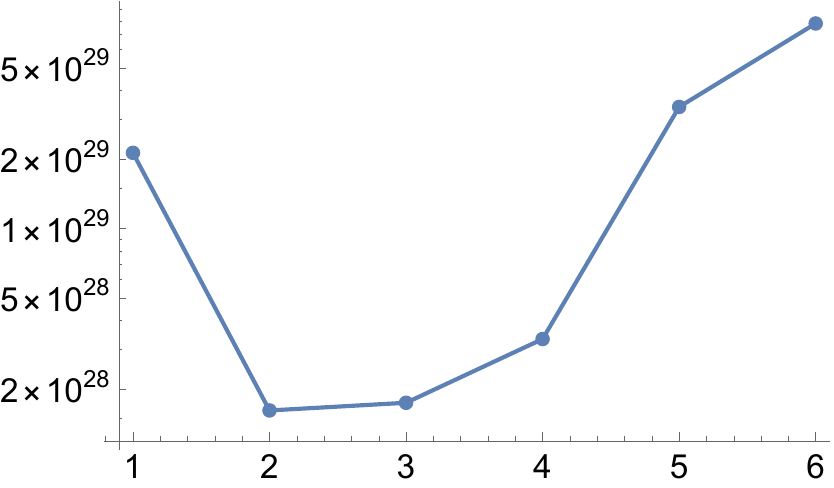}
            \put(50, -2){\makebox(0,0){\text{Runge-Kutta method order} p}} % x-axis label
            \put(-5, 33){\rotatebox{90}{\makebox(0,0){$N_{circ}(p)$}}} % y-axis label
        \end{overpic}
    \end{subfigure}
\vspace{0.1cm}
    \caption{Comparisons of the resource requirements for different RKM orders, for the parameters in Table~\ref{tab:estimates}, applicable to the option pricing application. In the first column, we have the RKM order $p$. In the second column is $N_{circ}$  of an algorithm that uses an RKM of order $p$ and in the third column the ratio of $N_{circ}$ of an algorithm that uses the Euler method $(p=1)$ with $N_{circ}$ of an algorithm that uses an RKM of order $p$. In the third, fourth, and fifth columns we show $N_r^{(\delta)}$, the number of time steps $ N_{\tau}^{(\delta)}$ and the number of different circuits. In the plot, we show $N_{circ}(p)$ plotted against $p$. The total number of circuit evaluations $N_{circ}$ is calculated according to Eq.~\eqref{eq: cost function with nvfactor}. }
    \label{fig:costdifferentp}
\end{figure}

\begin{figure}[htbp!]
    \centering
    \vspace{0cm} % Adjust vertical alignment
    \hspace{0cm} % Adjust horizontal alignment
    \begin{subfigure}[t]{0.9\textwidth} 
        \centering
        \begin{tabular}{|c|c|c|c|c|c|}
    \hline
    $p$ & $N_{circ}(p)$ & $N_{circ}(1)/N_{circ}(p)$ & $N_{r}^{(\delta)}$ & $N_{\tau}^{(\delta)}$ & \text{Circuits} \\
\hline
1  & $1.12 \times 10^{37}$ & $1$  & $1.15 \times 10^{25}$ & $9.56 \times 10^{8}$  & $9.80 \times 10^{11}$ \\ \hline
2  & $2.63 \times 10^{34}$ & $4.28 \times 10^{2}$  & $3.93 \times 10^{25}$ & $3.26 \times 10^{5}$  & $6.68 \times 10^{8}$  \\ \hline
3  & $6.33 \times 10^{33}$ & $1.78 \times 10^{3}$  & $9.57 \times 10^{25}$ & $2.15 \times 10^{4}$  & $6.61 \times 10^{7}$  \\ \hline
4  & $4.39 \times 10^{33}$ & $2.56 \times 10^{3}$  & $1.98 \times 10^{26}$ & $5.41 \times 10^{3}$  & $2.22 \times 10^{7}$  \\ \hline
5  & $1.00 \times 10^{34}$ & $1.12 \times 10^{3}$  & $6.78 \times 10^{26}$ & $2.40 \times 10^{3}$  & $1.48 \times 10^{7}$  \\ \hline
6  & $1.11 \times 10^{34}$ & $1.01 \times 10^{3}$  & $1.14 \times 10^{27}$ & $1.36 \times 10^{3}$  & $9.76 \times 10^{6}$  \\ \hline
7  & $2.60 \times 10^{34}$ & $4.33 \times 10^{2}$  & $3.06 \times 10^{27}$ & $9.22 \times 10^{2}$  & $8.50 \times 10^{6}$  \\ \hline
8  & $5.90 \times 10^{34}$ & $1.91 \times 10^{2}$  & $7.61 \times 10^{27}$ & $6.88 \times 10^{2}$  & $7.75 \times 10^{6}$  \\ \hline
9  & $1.33 \times 10^{35}$ & $84.69 $  & $1.82 \times 10^{28}$ & $5.47 \times 10^{2}$  & $7.29 \times 10^{6}$  \\ \hline
10 & $4.92 \times 10^{35}$ & $22.87$  & $6.48 \times 10^{28}$ & $4.63 \times 10^{2}$  & $7.59 \times 10^{6}$  \\
\hline
\end{tabular}
    \end{subfigure}

    \vspace{0.5cm} % Add vertical space between table and plot
    \begin{subfigure}[t]{0.9\textwidth} 
        \centering
        \begin{overpic}[scale=0.6,grid=false,tics=10]{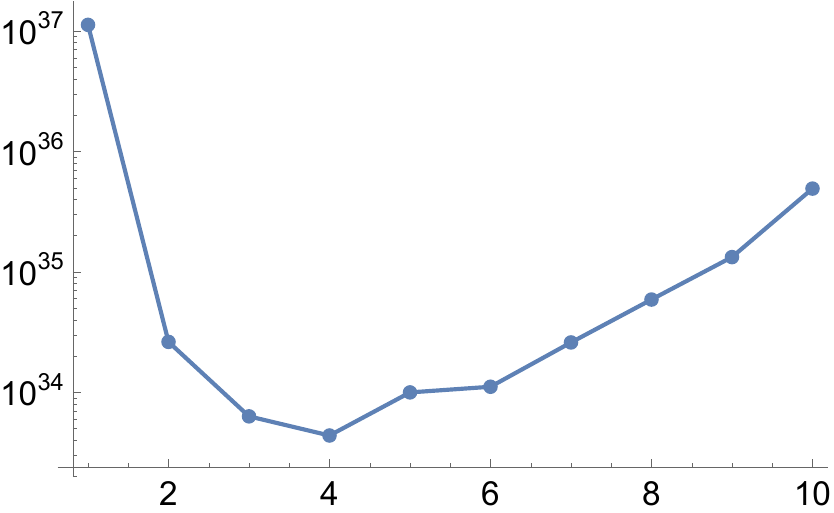}
            \put(50, -2){\makebox(0,0){\text{Runge-Kutta method order} p}} % x-axis label
            \put(-5,33){\rotatebox{90}{\makebox(0,0){$N_{circ}(p)$}}} % y-axis label
        \end{overpic}
    \end{subfigure}
\vspace{0.1cm}
    \caption{Comparison of $N_{circ}$ for different RKM orders, where the estimated parameters are fine-tuned to maximize the savings by using a higher order $p$ as given in Table~\ref{tab:estimates maximize gain}. In the first column, we have the Runge-Kutta method order $p$. In the second column is $N_{circ}$ of an algorithm that uses a Runge-Kutta method of order $p$ and in the third column the ratio of $N_{circ}$ of an algorithm that uses the Euler method $(p=1)$ with $N_{circ}$ of an algorithm that uses a Runge-Kutta method of order $p$. In the third, fourth, and fifth columns, we show the number of shots per circuit $N_{r}^{(\delta)}$, the number of time steps $N_{\tau}^{(\delta)}$, and the number of circuits. In the plot, we show $N_{circ}(p)$ plotted against $p$. The total number of circuit evaluations $N_{circ}$ is calculated according to Eq.~\eqref{eq: cost function with nvfactor}.}
    \label{fig:costdifferentp maximize gain}
\end{figure}

\begin{table}[htbp]
    \centering
    \begin{tabular}{|c|c|}
        \hline
        \textbf{Parameter} & \textbf{Estimate} \\
        \hline
        \( a_{max}\) & \(  1 \) \\
        \hline
        \( b_{max}\) & \( 0.5 \) \\
        \hline
        \( L_{fy} \) & \( 0.1 \) \\
        \hline
        \( T \) & \( 4 \) \\
        \hline
        \( K \) & \(  20 \) \\
        \hline
        \( L_{f\tau} \) & \(   15 \) \\
        \hline
        \( M \) & \(   60 \) \\
        \hline
        \( S \) & \(   1 \) \\
        \hline
        \( \Sigma \) & \(   3.4\times 10^8 \)\\
        \hline
         \( \epsilon_{target} \) & \(   0.001 \)\\
        \hline
    \end{tabular}
    \caption{Estimates of various parameters that optimize the savings by using a higher-order Runge-Kutta method instead of the Euler method. Used for the analysis in  Fig.~\ref{fig:costdifferentp maximize gain}.}
    \label{tab:estimates maximize gain}
\end{table}

\section{Conclusions}
\label{sec:conclusion}
In this work, we developed error and resource estimates of variational algorithms for solving differential equations based on Runge-Kutta methods. In particular, the estimates depend on different parameters that come from both the chosen method and the differential equation at hand, and depend on both the error from the Runge-Kutta method and the error that comes from shot noise in the noisy evaluations of the differential function. 
Additionally, we performed numerical simulations of the minimal resources required for both solving a simple ODE by using Runge-Kutta methods without shot noise and for solving a variational algorithm applied to a linear PDE from option pricing.
This shows that our method is not restricted to ODEs but can also be applied to solving partial or stochastic differential equations. Further, the algorithms we analyzed are based on the variational quantum algorithms for solving real and imaginary time evolution \cite{li2017efficient,mcardle2019variational}; our analysis can therefore be directly applied to them.
\par
Our results suggest that depending on certain parameters, such as the difference between initial and final time, Lipschitz constants of the differential function, and target error, the resource requirements change drastically. 
\par
Further, we show that for the particular ODE that we solve in the case without shot noise and a chosen target error of at most $\epsilon_{target}=0.001$, a fourth-order Runge-Kutta method is the most resource-efficient, while for the option pricing use case solved by a variational quantum algorithm and the same target error of $\epsilon_{target}=0.001$, the most resource-efficient method is a second-order Runge-Kutta method. 
By choosing these methods, we can minimize the total cost by a factor of $2.22\times 10^3$ for the first case and the total number of quantum circuit evaluations by a factor of  $13.18$ for the option pricing use case. 
For solving the option pricing PDE, even when using a Runge-Kutta method with the order of $p=2$, we estimated that the algorithm needs at least $1.62\times 10^{28}$ evaluations of quantum circuits in order to compute the option price at final time with a maximum error of $\epsilon_{target}=0.001$ in trace distance. 
In practical scenarios, the state has to be read out, and therefore the total resource requirements will be even higher. 
\par However, our results are worst-case estimates that guarantee staying within the target error, while in practice, one might expect to achieve them with lower resource requirements for sufficiently well-behaved functions. 
In particular, Lemma~\ref{lemma:liubovs sensitivity} gives an upper bound for errors induced by solving linear systems. 
But even for ill-conditioned problems, a smaller error is possible in practice (see for example Ref.~\cite{datta2004numerical}).
Also, the estimates in Secs.~\ref{sec:estimations on toy model} and~\ref{sec: numerics} are upper bounds that can in general be tightened for specific problems at hand.
 
The PDE chosen in our use case, the Black-Scholes equation, can be solved analytically, so only more complicated dynamics are of practical interest. 
Considering this and the fact that a requirement of at least $3.87\times 10^{22}$ evaluations for each of the $4.19\times 10^{5}$ quantum circuits (of which only $10^3$ can be run in parallel, see Sec.~\ref{sec: Variational Quantum Simulator}) is unrealistic on current quantum computing platforms, our results are rather pessimistic for practical implementations of this algorithm in option pricing. 
As discussed in the previous section, this high requirement of quantum circuit evaluations mainly stems from the error-propagation of inverting a matrix of estimated observables.
This leads us to question the general feasibility of the underlying quantum algorithm.
However, related approaches as proposed in Refs.~\cite{barison2021efficient} and~\cite{benedetti2021hardware} which do not require matrix inversion, might therefore require a much lower number of quantum circuit evaluations.
\par 
We show that by tuning the parameters in the variational algorithm and therefore reaching realms outside of this use case, we can get a saving factor of  $2.56 \times 10^{3}$ by choosing a Runge-Kutta method of order $p=4$ instead of the Euler method. This suggests that after careful analysis of the parameters, one can choose a Runge-Kutta method that minimizes the resource requirements.
\par 
Our analysis had a number of limitations, which however likely do not make matters better. We do not analyze possible stability issues of the differential equations, because it proves very challenging to include them in the estimation of the variational algorithm and the option pricing use case. However, for a full picture of the error and resource analysis, the stability of methods and DEs has to be considered.
We neglected possible representation errors that stem from the quantum circuits not being able to approximate a state that encodes the option price and which can capture the dynamics of the variational parameters satisfyingly. In practical scenarios, these have to be taken into account and can possibly be estimated depending on the chosen quantum computing platform (see for example Ref.~\cite{gacon2024variational}). Further, we assumed circuit error such as gate infidelity,
bias and SPAM errors to be negligible, as well as errors introduced to potentially necessary matrix regularization of the matrix defined in Eq.~\eqref{eq: A evaluation}. 
However, as long as when adding those sources of errors, our error bounds (for example in Eq.~\eqref{eq:upper bound f error}) still hold, our results can be applied to these scenarios as well.

\par It might be possible to tighten our estimates in a few different ways. Recent works~\cite{kolotouros2024accelerating,anuar2024operatorprojectedvariationalquantumimaginary} showed how it is possible to decrease the number of state preparations for estimations of the matrices $A$ from scaling as $\Theta(N_V^2)$ to $\Theta(N_V)$ at the cost of accuracy of the matrix entries. Combining these results with our analysis might further decrease the total resource requirements.
Several bounds in our analysis are not tight and could show to be overly pessimistic in real scenarios, like the local truncation error of the Runge-Kutta method (Thm.~\ref{thm:local_truncation_RKE}) and the bounds and estimates on the shot noise (Lemma~\ref{lemma: shot noise_2} and Sec.~\ref{subsubsec: Parameter estimates and Numerical Analysis quantum}).
Further, the resource requirements can possibly be decreased by using linear multistep methods that only use one new stage per time step and reuse evaluations from previous time steps. Adapting our analysis to these methods can give new insights into a comparison of Runge-Kutta methods and linear multistep methods.
\par Since the shot noise error introduced for the analysis in Sec.~\ref{sec: Analysis with Shot Noise Scaling} is Gaussian noise, it might give further insight to formulate the differential equations as stochastic differential equations and to analyze the error and resource requirements within this framework.
\par 
It might be promising to apply this framework to use cases that use Runge-Kutta methods and have similar error sources, such as quantum algorithms for solving other differential equations or to classical algorithms that are using noisy data, which so far has been barely examined. Also, it might be possible to apply this analysis to the training of neural networks ~\cite{xie2020linear}.

\section{Acknowledgements}
The authors thank Patrick Emonts, Onur Danacı, Hao Wang, Xavier Bonet-Monroig and Christa Zoufal for useful discussions during the project.

JT, LM, VD and DD acknowledge the support received by the Dutch National Growth Fund (NGF), as part of the Quantum Delta NL programme.
JT acknowledges the support received from the European Union’s Horizon Europe research and innovation programme through the ERC StG FINE-TEA-SQUAD (Grant No. 101040729).'
This publication is part of the project Divide \& Quantum (with project number 1389.20.241) of the research programme NWA-ORC which is (partly) financed by the Dutch Research Council (NWO).
VD acknowledges by the Dutch Research Council
(NWO/OCW), as part of the Quantum Software Consortium programme (project number 024.003.037).
VD acknowledges co-funding by the European Union (ERC CoG, BeMAIQuantum, 101124342).
LM was  supported by the Netherlands Organisation for Scientific Research (NWO/OCW), as part of the Quantum Software Consortium program (project number 024.003.037 / 3368).

The views and opinions expressed here are solely those of the authors and do not necessarily reflect those of the funding institutions. Neither of the funding institution can be held responsible for them. 
\appendix

\section{Proof of Theorems~\ref{thm:classicalODEerror} and~\ref{thm:quantumODEerror}}
\label{Ax:Analysis with shot noise}
\begin{proof}[Proof]
Let us denote the RKM calculated ${\bm{\theta}}_n$ with an error-carrying $\hat{f}$ at time step $n$ as
$\hat{\bm{\theta}}_n$.
Calculating the LTE in $\hat{\bm{\theta}}_n$, we assume that it is calculated from a noiseless $\bm{\theta}(\tau_{n-1})$. For an $s$-stage RKM $\hat{\bm{\theta}}_n$ is calculated by:
\begin{align}
    \hat{\bm{\theta}}_n=\bm{\theta}(\tau_{n-1})+\Delta \tau\sum_{i=1}^{s}b_i\hat{k}_i\left(\tau_{n-1};\bm{\theta}(\tau_{n-1});\{\hat{k}_m\}_{m=1}^{i-1}\right)\ ,
\end{align}
where we write
\begin{align}
&\hat{k}_1(\tau_{n-1};\bm{\theta}(\tau_{n-1})):=\hat{f}\left(\tau_{n-1};\bm{\theta}(\tau_{n-1})\right),\quad\text{and}\\
    &\hat{k}_i\left(\tau_{n-1};\bm{\theta}(\tau_{n-1});\{\hat{k}_m\}_{m=1}^{i-1}\right):=\hat{f}\left(\tau_{n-1}+c_i \Delta\tau;\bm{\theta}(\tau_{n-1})+\Delta\tau\sum_{m=1}^{i-1}a_{i,m}\hat{k}_m \right)\quad\text{for all } i\geq 2,
\end{align}
and where we used the abbreviation $\hat{k}_m$ by dropping the arguments.

Thus, the LTE in the presence of noise is:
\begin{align}
    \hat{\ell_n}:=\bm{\theta}(\tau_n)-\hat{\bm{\theta}}_n=\bm{\theta}(\tau_n)-\bm{\theta}(\tau_{n-1})-\Delta \tau\sum_{i=1}^{s}b_i\hat{k}_i\left(\tau_{n-1};\bm{\theta}(\tau_{n-1});\{\hat{k}_m\}_{m=1}^{i-1}\right)
\end{align}
We can immediately write
\begin{align}
    \hat{\ell}_n&\leq \bm{\theta}(\tau_n)-\bm{\theta}_n+\bm{\theta}_n-\hat{\bm{\theta}}_n
    \leq \ell_n+\bm{\theta}_n-\hat{\bm{\theta}}_n\ ,
\end{align}
where the noiseless LTE $\ell_n$ can be upper bounded by the bound in Thm.~\ref{thm:local_truncation_RKE}.
\par Let us define the global truncation error of an RKM with $s$ stages and $n$ time steps in the presence of noise as:
\begin{align}
    \hat{e}_n:= 
    \bm{\theta}(\tau_n)-\hat{\bm{\theta}}_0-\Delta \tau\sum_{r=1}^n\sum_{i=1}^{s}b_i \doublehat{k_i}\left(\tau_{r-1}; \hat{\bm{\theta}}_{r-1}; \{\doublehat{k}_m\}_{m=1}^{i-1}\right)\ ,
\end{align}
where we write \begin{align}
&\doublehat{k}_1\left(\tau_{r-1}; \hat{\bm{\theta}}_{r-1}\right):=\doublehat{f}\left(\tau_{r-1}; \hat{\bm{\theta}}_{r-1}\right),\quad\text{and}\\
    &\doublehat{k_i}\left(\tau_{r-1}; \hat{\bm{\theta}}_{r-1}; \{\doublehat{k}_m\}_{m=1}^{i-1}\right):=\doublehat{f}\left(\tau_{r-1}+c_i \Delta\tau;\hat{\bm{\theta}}_{r-1}+\Delta\tau\sum_{m=1}^{i-1}a_{i,m}\doublehat{k}_m \right)\quad\text{for all } i\geq 2,
    \end{align}
    and again use the abbreviation $\doublehat{k}_m$ by dropping the arguments. We write the superscript $\doublehat{f}$ instead of $\hat{f}$ in order to distinguish the noisy evaluation with noisy inputs with the noisy evaluation with noiseless inputs. They both however carry the error $\delta$ compared to the noise-free case, as written in Eq.~\eqref{eq: error in evaluating f}.
Recursively, we get:
\begin{align}
    \hat{e}_{n+1}-\hat{e}_n
    =\hat{\ell}_{n+1}+\Delta \tau\sum_{i=1}^{s}b_i\left( \hat{k}_i\left(\tau_{n};\bm{\theta}(\tau_{n});\{\hat{k}_m\}_{m=1}^{i-1}\right)-\doublehat{k_i}\left(\tau_{n}; \hat{\bm{\theta}}_{n}; \{\doublehat{k}_m\}_{m=1}^{i-1}\right)\right)\ .
\end{align}
Taking the absolute values from the latter equation, we get
\begin{align}
    \label{eq:generalerrorestimatewithshotnoise}
    \|\hat{e}_{n+1}\|\leq\|\hat{e}_{n}\|+\|\hat{\ell}_{n+1}\|+\Delta \tau\sum_{i=1}^{s}|b_i|\left\|\hat{k}_i\left(\tau_{n};\bm{\theta}(\tau_{n});\{\hat{k}_m\}_{m=1}^{i-1}\right)-\doublehat{k_i}\left(\tau_{n}; \hat{\bm{\theta}}_{n}; \{\doublehat{k}_m\}_{m=1}^{i-1}\right)\right\|\  .
\end{align}
We assume that the function $f(\tau,\bm{\theta}(\tau))$  satisfies the Lipschitz condition:
\begin{align}\label{1149}
    \left\|f(\tau_n,\bm{\theta}(\tau_n))-f(\tau_n,\hat{\bm{\theta}}_n)\right\|
    \leq & L_{fy}\|\bm{\theta}(\tau_n)-\hat{\bm{\theta}}_n\|\ , \end{align} and that the noise evaluations are upper bounded by
    \begin{align}
        \|\hat{k}_i-k_i\|\leq \delta\\
        \|\doublehat{k}_i-k_i\|\leq \delta\ .
    \end{align}
    We introduce the notation
\begin{align}
    S_i:= \left\|\hat{k}_i\left(\tau_{n};\bm{\theta}(\tau_{n});\{\hat{k}_m\}_{m=1}^{i-1}\right)-\doublehat{k_i}\left(\tau_{n}; \hat{\bm{\theta}}_{n}; \{\doublehat{k}_m\}_{m=1}^{i-1}\right)\right\|.
    \end{align}
Using the triangle inequality, we get
\begin{align*}
    S_1=\left\|\hat{k}_1\left(\tau_{n};\bm{\theta}(\tau_{n}))\right)-\doublehat{k}_1\left(\tau_{n}; \hat{\bm{\theta}}_{n}\right)\right\|
    &\leq  \left\|\hat{k}_1\left(\tau_{n};\bm{\theta}(\tau_{n}))\right)-k_1\left(\tau_{n};\bm{\theta}(\tau_{n}))\right)\right\|
+\left\|k_1\left(\tau_{n}; \hat{\bm{\theta}}_{n}\right)-\doublehat{k}_1\left(\tau_{n}; \hat{\bm{\theta}}_{n}\right)\right\|\\
&+\left\|k_1\left(\tau_{n};\bm{\theta}(\tau_{n}))\right)-k_1\left(\tau_{n}; \hat{\bm{\theta}}_{n}\right)\right\|
    \leq  
    2\delta+L_{fy}\|\hat{e}_{n}\|.
\end{align*}

Again using the triangle inequality, we get for $S_i$, $i\geq 2$:
\begin{align}
    &S_i= \left\|\hat{k}_i\left(\tau_{n};\bm{\theta}(\tau_{n});\{\hat{k}_m\}_{m=1}^{i-1}\right)-\doublehat{k_i}\left(\tau_{n}; \hat{\bm{\theta}}_{n}; \{\doublehat{k}_m\}_{m=1}^{i-1}\right)\right\|\\
    \leq&\left\|\hat{k}_i\left(\tau_{n};\bm{\theta}(\tau_{n});\{\hat{k}_m\}_{m=1}^{i-1}\right)-k_i\left(\tau_{n};\bm{\theta}(\tau_{n});\{\hat{k}_m\}_{m=1}^{i-1}\right)\right\|\\
    &+\left\|k_i\left(\tau_{n}; \hat{\bm{\theta}}_{n}; \{\doublehat{k}_m\}_{m=1}^{i-1}\right)-\doublehat{k_i}\left(\tau_{n}; \hat{\bm{\theta}}_{n}; \{\doublehat{k}_m\}_{m=1}^{i-1}\right)\right\|\\
    &+\left\|k_i\left(\tau_{n};\bm{\theta}(\tau_{n});\{\hat{k}_m\}_{m=1}^{i-1}\right)-k_i\left(\tau_{n}; \hat{\bm{\theta}}_{n}; \{\doublehat{k}_m\}_{m=1}^{i-1}\right)\right\|\ \\
    \leq & 2\delta+L_{fy}\left\|\bm{\theta}(\tau_n)-\hat{\bm{\theta}}_n\right\|+\Delta \tau\sum_{m=1}^{i-1}L_{fy} |a_{im}| \left\|\hat{k}_m-\doublehat{k}_m\right\|\\
    \leq & 2\delta+L_{fy}\|\hat{e}_{n}\|+\Delta \tau\sum_{m=1}^{i-1}L_{fy} |a_{im}| S_m\\
    \leq & S_1+\Delta \tau\sum_{m=1}^{i-1}L_{fy} |a_{im}| S_m\ .
\end{align}
Then, we can write the upper bound of the error of the $s$ order RKM as
\begin{align}\label{eq:generalordermethodglobalestimat}
\|\hat{e}_{n+1}\|&\leq
\|\hat{e}_{n}\|+\|\hat{\ell}_{n+1}\|+\Delta \tau\sum_{i=1}^{s}|b_i|S_i\ ,
\end{align}
where from the above analysis, we obtained the following recursion:
\begin{align}
    S_1&\leq 2\delta+L_{fy}\|\hat{e}_{n}\|\\
    S_{i}&\leq S_1+ \Delta \tau\sum_{m=1}^{i-1}L_{fy}|a_{im}|S_m,\quad \forall i\geq 2.
\end{align}
Let us use that $|a_{ij}|$ is upper bounded with $a_{max}:=\max_{i,j}|a_{i,j}|$. Then we can write:
\begin{align}
    S_1&\leq 2\delta+L_{fy}\|\hat{e}_{n}\|\\
    S_i&\leq S_1+ \Delta \tau L_{fy}a_{max}\sum_{m=1}^{i-1}S_m,\quad \forall i\geq 2.
\end{align}
A proof by induction shows that for all $i\geq 1$,
\begin{align}
    S_i&\leq S_1\left(1+\Delta \tau L_{fy}a_{max}\right)^{i-1}.
\end{align}
This result together with Eq.~\eqref{eq:generalordermethodglobalestimat} gives us the error estimate:
\begin{align}
    \|\hat{e}_{n+1}\|
    &\leq \|\hat{e}_{n}\|+\|\hat{\ell}_{n+1}\|+\Delta \tau S_1\sum_{i=1}^{s}|b_i| \left(1+\Delta \tau L_{fy}a_{max}\right)^{i-1}\\
    &\leq \|\hat{e}_{n}\|+\|\hat{\ell}_{n+1}\|+\Delta \tau S_1b_{max}\frac{\left(1+\Delta \tau L_{fy}a_{max}\right)^{s}-1}{L_{fy}\Delta \tau a_{max}}\\
&\leq\|\hat{e}_{n}\|+\left(\max_n\|\hat{\ell}_{n}\|+\left(2\delta+L_{fy}\|\hat{e}_{n}\|\right)b_{max}\frac{\left(1+\Delta \tau L_{fy}a_{max}\right)^{s}-1}{L_{fy} a_{max}}\right) \\
    &\leq \alpha\|\hat{e}_{n}\|+\beta, \quad\text{for all} s\geq 1,
    \end{align}
    where we used the notation
    \begin{align}
    &\alpha =1+F(n+1,s),\quad \beta = \frac{2\delta}{L_{fy}}F(n+1,s) + \max_n\|\hat{\ell}_{n}\|,\\
&    F(n+1,s):=  \frac{b_{max}}{a_{max}} \left(\left(1+ \frac{\Theta}{n+1}\right)^{s}-1\right),\quad \Theta=L_{fy}a_{max}T,
\end{align}
an the fact that $\Delta \tau =T/(n+1)$ holds. 
. Because $\|\hat{e}_{0}\|= 0$, 
 we get
\begin{align}\label{1529}
    \|\hat{e}_{n+1}\|&\leq \frac{\alpha^{n+1}-1}{\alpha-1}\beta.
\end{align}
The error estimate of the LTE goes analogous to the error estimate of the global truncation error. We can write 
\begin{align}
    \|\hat{\ell}_n\|\leq\|\ell_n\|+\Delta \tau\sum_{i=1}^s|b_i|T_i 
\end{align}
with the recursive relation
\begin{align}
    T_1:=\delta,\quad 
    T_i:=T_1+ \Delta \tau\sum_{m=1}^{i-1}L_{fy}|a_{im}|T_m.
\end{align}
Thus, we get the bounds 
\begin{align}
&\|\hat{\ell}_n\|\leq\|\ell_n\|+ \frac{\delta}{ L_{fy}} F(n+1,s),\quad\text{for }s\geq1.
\end{align}
Using these upper bounds, we can rewrite Eq.~\eqref{1529} as follows
\begin{align}
    \|\hat{e}_{n+1}\|&\leq \frac{\alpha^{n+1}-1}{\alpha-1}\beta^* \ ,
\end{align}
where
\begin{align}
    \beta^* = \frac{3\delta}{ L_{fy}} F(n+1,s) + \max_{n}\|\ell_{n}\|
\end{align}
According to the Theorem~\ref{thm:local_truncation_RKE}, we can bound the LTE $\ell_n$ and get after $N_t:=n+1$ steps the following error:
\begin{align}
    \|\hat{e}_{N_{\tau}}\|&\leq 
\frac{(1+ F(N_{\tau},s))^{N_{\tau}}-1}{F(N_{\tau},s)}\left(\frac{3\delta}{Lfy}F(N_{\tau},s)+\left(\frac{T}{N_{\tau}}\right)^{p+1}  K  L^p_{f\tau}   M\right).     
\end{align}
That ends the proof for Thm.~\ref{thm:quantumODEerror}, for which $f(\tau,\bm{\theta}(\tau))=f(\bm{\theta}(\tau))$.
One can see that noise free case corresponds to $\delta=0$, for which $y(\tau_n)$ corresponds to $\bm{\theta}(\tau_{n})$ and $y_n$ to $\bm{\theta}_n$, which proofs Thm.~\ref{thm:classicalODEerror}.
\end{proof}

\section{Proof of Theorem~\ref{thm:classicalODEresources}}
\label{app:proof of classical resources}
\begin{proof}[Proof]

Let us denote the upper bound of the noise free $\epsilon_{ODE}^{(\delta)}$ as
 \begin{align}\label{1232}
  \epsilon_{target}^{(0)}:= \frac{(1+F(N_{\tau},s))^{N_{\tau}}-1}{F(N_{\tau},s)}\left(\frac{T}{N_{\tau}}\right)^{p+1}  K  L^p_{f\tau}   M \ ,
\end{align}

To find $N_{\tau}$ we can solve the latter equation numerically. 
However, under the reasonable assumption that $\Theta<<N_{\tau}$, we can use the following approximations
\begin{align}
    &\left(\frac{\Theta}{N_{\tau}}+1\right)^{s}= 1+\frac{s \Theta}{N_{\tau}}+O\left(\frac{\Theta^2}{N_{\tau}^2}\right),\\
    &\left(\frac{b_{max} (s \Theta)}{a_{max} N_{\tau}}+1\right)^{N_{\tau}}\approx\exp \left(\frac{b_{max} s \Theta}{a_{max}}\right)\label{1650},\end{align}
    to rewrite  
    \begin{align}
        F(N_{\tau},s)\approx \frac{b_{max}}{a_{max}}\frac{s \Theta}{N_{\tau}}, \quad (1+F(N_{\tau},s))^{N_{\tau}}\approx \exp \left(\frac{b_{max} s \Theta}{a_{max}}\right)
    \end{align}
    Using this, we rewrite    Eq.~\eqref{1232} in a way to obtain an approximate solution  $N_{\tau}^{(0)}$ given by Eq.~\eqref{1245}. 
\end{proof}

\section{Proof of Theorem~\ref{thm:quantumODEresources}}
\label{appendix:proof shot noise resources}
\begin{proof}[Proof]
We model the shot noise with single shot variance $\Sigma$ and the number of shots $N_r$, so  $ \delta=\Sigma/\sqrt{N_r}$ holds. Solving the resulting Eq.~\eqref{eq:bound on target error for quantum resource thm}
 for $N_r$ we get:
\begin{align}
    N_{r}^{(\delta)}:= \frac{9\Sigma^2}{L_{fy}^2}   
   \left(\frac{\epsilon_{target}^{(\delta)}}{\left(1+ F(N_{\tau},s)\right)^{N_{\tau}}-1}-\left(\frac{T}{N_{\tau}}\right)^{p+1} \frac{KL^p_{f\tau}  M}{F(N_{\tau},s)}\right)^{-2}\ ,
\end{align}

Substituting it in the cost function~\eqref{eq:cost function with shot noise}  we get:
\begin{align}\label{eq:costfunction}
    \mathcal{C}(N_{\tau}, N_{r}^{(\delta)},s,p)
    &=\frac{9\Sigma^2p}{L_{fy}^2}   N_{\tau}^{2p+3}
\left(\frac{F(N_{\tau},s)(\left(1+ F(N_{\tau},s)\right)^{N_{\tau}}-1)}{\epsilon_{target}^{(\delta)}F(N_{\tau},s)N_{\tau}^{p+1}-T^{p+1} KL^p_{f\tau}  M (\left(1+ F(N_{\tau},s)\right)^{N_{\tau}}-1)}\right)^{2},
\end{align}
where 
\begin{align}
    \epsilon_{target}^{(\delta)}\neq \frac{T^{p+1} KL^p_{f\tau}  M (\left(1+ F(N_{\tau},s)\right)^{N_{\tau}}-1)}{F(N_{\tau},s)N_{\tau}^{p+1}}= \epsilon_{target}^{(0)},
\end{align}
that is always true.
We want to find the optimal value for $N_{\tau}$ that minimizes the latter cost. To this end, we set the derivative of the cost function with respect to $N_{\tau}$  to zero. The equation we get in nonlinear and can be solved numerically.

However, we can use the approximation in Eq.~\eqref{1650}
and the expression
\begin{align}
    N_{\tau}\log \left(\frac{b_{max}s \Theta}{a_{max}N_{\tau}}+1\right)=\log \left(\frac{b_{max}s \Theta}{a_{max}N_{\tau}}+1\right)^{N_{\tau}}\approx \left(\frac{b_{max}s \Theta}{a_{max}}\right).
\end{align}
Those two approximations give the following simplified equation
\begin{align}
&\frac{\epsilon_{target}^{(\delta)}\left(a_{max}N_{\tau} (N_{\tau}+\Theta)-b_{max}s \Theta\left(2 N_{\tau}^2 e^{\frac{b_{max}s \Theta}{a_{max}}}+N_{\tau} \left(2 s \Theta e^{\frac{b_{max}s \Theta}{a_{max}}}-1\right)-\Theta\right)\right)}{(N_{\tau}+\Theta) (a_{max}N_{\tau}+b_{max}s \Theta) \left(e^{\frac{b_{max}s \Theta}{a_{max}}}-1\right)}\\
&+\frac{2 \epsilon_{target}^{(\delta)}(b_{max}s \Theta) e^{\frac{b_{max}s \Theta}{a_{max}}}}{a_{max}\left(e^{\frac{b_{max}s \Theta}{a_{max}}}-1\right)}-\frac{a_{max}K M L_{f\tau}^p \left(\frac{T}{N_{\tau}}\right)^{p+1} (2 N_{\tau} p+N_{\tau}-2 (s-1) \Theta)}{b_{max}s \Theta}=0\ .
\end{align}
Neglecting  some of the terms due to the fact that $\Theta<<N_{\tau}$
we simplify the equation to the form that it can be solved analytically:
\begin{align} 
    N_{\tau}^{(\delta)}=T L_{f\tau}\left(\frac{\left(e^{b_{max}sL_{fy}T}-1\right)KM(2p+1)}{sb_{max}L_{fy}\epsilon_{target}^{(\delta)}}\right)^{1/p}.
\end{align}
That ends the proof.
\end{proof}

\section{McLachlan's variational principle}
\label{appendix: McLachlan's principle}
\begin{thm}
    McLachlan's variational principle ~\cite{mclachlan1964variational} applied to the imaginary time evolution of a state $\ket{\psi(\bm{\theta}(\tau))}$, is given by:
    \begin{align}
        \delta\|(d/d\tau + \mathcal{H})\ket{\psi(\theta(\tau))}\|=0\ . \label{eq:McLachlan}
    \end{align}
    Assuming ${\theta}$ to be real, it is solved by the following ordinary differential equations ~\cite{yuan2019theory}:
    \begin{align}
        \sum_jA_{ij}\dot{\theta_j}=C_{i} , \label{eq:parameter_ODE}
    \end{align}
    where the matrix elements are
    \begin{align}
    A_{ij}=\mathfrak{Re}\left(\frac{\partial\bra{\psi(\theta(\tau))}}{\partial\theta_i}\frac{\partial\ket{\psi(\theta(\tau))}}{\partial\theta_j}\right),\quad \text{ and} \quad 
        C_i=\mathfrak{Re}\left(-\frac{\partial\bra{\psi(\theta(\tau))}}{\partial\theta_i}\mathcal{H}\ket{\psi(\theta(\tau))}\right) .\label{eq:C}
    \end{align}
    \end{thm}
    \begin{proof}
    Note that the variation in Eq.~\eqref{eq:McLachlan} is equivalent to the variation of its square $\delta\|(d/d\tau+\mathcal{H})\ket{\psi(\theta(\tau))}\|^2=0$. 
    The derivative can be written as follows
    \begin{align}
      \frac{\partial \ket{\psi(\theta(\tau))}}{\partial \tau}=  \frac{\partial \ket{\psi(\theta(\tau))}}{\partial \theta_i}\frac{\partial \theta_i}{\partial\tau}=\frac{\partial \ket{\psi(\theta(\tau))}}{\partial \theta_i}\dot{\theta}_i.
    \end{align}
    Next, we expand the following:
    \begin{align}
        \|(d/d\tau+\mathcal{H})\ket{\psi(\theta(\tau))}\|^2   =&\big((d/d\tau+\mathcal{H})\ket{\psi(\theta(\tau))}\big)^{\dagger}\big((d/d\tau+\mathcal{H})\ket{\psi(\theta(\tau))}\big)\\\nonumber
        =&\sum_{ij}\frac{\partial\bra{\psi(\theta(\tau))}}{\partial\theta_i}\frac{\partial\ket{\psi(\theta(\tau))}}{\partial\theta_j}\dot{\theta_i^*}\dot{\theta_j}+\sum_{i}\frac{\partial\bra{\psi(\theta(\tau))}}{\partial\theta_i}(\mathcal{H})\ket{\psi(\theta(\tau))}\dot{\theta^*_i} \\\nonumber
        &+\sum_{i}\bra{\psi(\theta(\tau))}\mathcal{H}\frac{\partial\ket{\psi(\theta(\tau))}}{\partial\theta_i}\dot{\theta_i}+\bra{\psi(\theta(\tau))}\mathcal{H}^2\ket{\psi(\theta(\tau))}\ .
    \end{align}
    Thus, the variation of this term with respect to $\dot{\theta_i}$ yields:
    \begin{align}
        \delta\|(d/d\tau+\mathcal{H})\ket{\psi(\theta(\tau))}\|^2
        =&\left(\sum_{ij}\frac{\partial\bra{\psi(\theta(\tau))}}{\partial\theta_i}\frac{\partial\ket{\psi(\theta(\tau))}}{\partial\theta_j}\dot{\theta_j}+\frac{\partial\bra{\psi(\theta(\tau))}}{\partial\theta_i}\mathcal{H}\ket{\psi(\theta(\tau))}\right)\delta\dot{\theta_i^*} \\\nonumber
        &+\left(\sum_{ij}\frac{\partial\bra{\psi(\theta(\tau))}}{\partial\theta_j}\frac{\partial\ket{\psi(\theta(\tau))}}{\partial\theta_i}\dot{\theta_j^*}+\bra{\psi(\theta(\tau))}\mathcal{H}\frac{\partial\ket{\psi(\theta(\tau))}}{\partial\theta_i}\right)\delta\dot{\theta_i}\ .
    \end{align}
    The variational principle is satisfied if latter equation is equal to zero, hence when following expression holds:
    \begin{align}
        \sum_j\frac{\partial\bra{\psi(\theta(\tau))}}{\partial\theta_i}\frac{\partial\ket{\psi(\theta(\tau))}}{\partial\theta_j}\dot{\theta_j}=-\frac{\partial\bra{\psi(\theta(\tau))}}{\partial\theta_i}\mathcal{H}\ket{\psi(\theta(\tau))} \ .
    \end{align}
    If $\dot{\theta_i}$ is real, the variation changes to 
    \begin{align}
        &\delta\|(d/d\tau+\mathcal{H})\ket{\psi(\theta(\tau))}\|^2\\\nonumber
        =&\sum_j\left(\frac{\partial\bra{\psi(\theta(\tau))}}{\partial\theta_i}\frac{\partial\ket{\psi(\theta(\tau))}}{\partial\theta_j}+\frac{\partial\bra{\psi(\theta(\tau))}}{\partial\theta_j}\frac{\partial\ket{\psi(\theta(\tau))}}{\partial\theta_i}\right)\dot{\theta_j}\delta\dot{\theta_i} \\\nonumber
        &+\left(\frac{\partial\bra{\psi(\theta(\tau))}}{\partial\theta_i}\mathcal{H}\ket{\psi(\theta(\tau))}+\bra{\psi(\theta(\tau))}\mathcal{H}\frac{\partial\ket{\psi(\theta(\tau))}}{\partial\theta_i}\right)\delta\dot{\theta_i}\ .
    \end{align}
    This is equal to zero when the following holds:
    \begin{align}
    &\sum_j\left(\frac{\partial\bra{\psi(\theta(\tau))}}{\partial\theta_i}\frac{\partial\ket{\psi(\theta(\tau))}}{\partial\theta_j}+\frac{\partial\bra{\psi(\theta(\tau))}}{\partial\theta_j}\frac{\partial\ket{\psi(\theta(\tau))}}{\partial\theta_i}\right)\dot{\theta_j} \\\nonumber
        &=-\left(\frac{\partial\bra{\psi(\theta(\tau))}}{\partial\theta_i}\mathcal{H}\ket{\psi(\theta(\tau))}+\bra{\psi(\theta(\tau))}\mathcal{H}\frac{\partial\ket{\psi(\theta(\tau))}}{\partial\theta_i}\right)\ ,
    \end{align}
    which is equivalent to Eq.~\eqref{eq:parameter_ODE}.
    \end{proof}

\section{Shot noise estimates}
\label{appendix:Shot noise estimates}
   \begin{thm}(Shot noise error for evaluating $A$ defined in Eq.~\eqref{eq: A evaluation})
   
    \label{thm:Shot noise error for evaluating A}
    For $N_r$ evaluations of each of the circuits that calculate the matrix elements of $A$, we get the following bound for the probability:
    \begin{align}
        P\left(\|A-\hat{A}_{S}\|<\frac{\|\{\sigma_{k,l}\}_{k,l=1}^{N_V}\|}{\sqrt{N_r  \eta}}\right)> 1-\eta\ ,
    \end{align}
    where $0<\eta\leq 1$ and the elements of the standard deviation matrix are 
    \begin{align}\label{1445}
            \sigma_{k,l}=\sqrt{\sum_{i,j=1}^{N_d}|f_{k,i}^*f_{l,j}|^2}\ .
        \end{align}
    \end{thm}
    \begin{proof}
    For a random matrix $A$ with finite non-zero variance matrix $\sigma^2$ and expectation values matrix $\hat{A}_{S}$ the multi-dimensional Chebyshev's inequality 
        \begin{align}
            P(\|A-\hat{A}_{S}\|\geq k\|\sigma\|)\leq \frac{1}{k^2}\ ,
        \end{align}
        holds, for any real number $k>0$.
    If we measure each circuit $N_r$ times, the mean value is calculated as
        \begin{align}
            \hat{A}_{S}=\frac{1}{N_r}\sum_{m=1}^{N_r} \hat{A}_m\ ,
        \end{align}
     where each $ \hat{A}_m$ is the matrix calculated by evaluating each circuit one time.   
    The  $\sigma_{k,l}(\hat{A}_m)$ is the standard deviation  for each matrix element of $\hat{A}_m$, and the total standard deviation for each element in $\hat{A}_S$ is given by
        \begin{align}
            \sigma_{k,l}(\hat{A}_S)=\frac{\sigma_{k,l}(\hat{A}_m)}{\sqrt{N_r}}\ .
        \end{align}
        These are the elements that form the matrix $\sigma=\{\sigma_{k,l}(\hat{A}_m)\}^{N_V}_{k,l=1}$.
        Defining $\eta:=1/k^2$, we get therefore
        \begin{align}
           P\left(\|A-\hat{A}_{S}\|<\frac{\|\sigma\|}{\sqrt{N_r  \eta}}\right)> 1-\eta \ .
        \end{align}
        We can further bound the elements of $\sigma$ in the following way. Each matrix element $\sigma_{k,l}(\hat{A}_S)$ is in general  evaluated by several circuits. Let us bound the standard deviation of each single circuit  by the maximum standard deviation $1$, since the eigenvalues of the Pauli-$X$ matrix are $\{1,-1\}$. We obtain:
        \begin{align}
            \sigma_{k,l}(\hat{A}_m)\leq \sigma_{k,l}=\sqrt{\sum\limits_{i,j=1}^{N_d}|f_{k,i}^*f_{l,j}|^2}.
        \end{align}
        That ends the proof.
    \end{proof}

    \begin{thm}(Shot noise error for evaluating $C$ defined in Eq.~\eqref{eq: A evaluation})\\
    \label{thm:Shot noise error for evaluating C}
    For $N_r$ evaluations of each of the circuits that calculate the elements of $C$, we get the following bound:
    \begin{align}
        P\left(\|C-\hat{C}_{S}\|<\frac{\|\{\sigma_{k}\}_{k=1}^{N_V}\|}{\sqrt{N_r  \eta}}\right)> 1-\eta\ ,
    \end{align}
    where $0<\eta\leq 1$ and \begin{align}\label{1444}
            \sigma_{k}=\sqrt{\sum_{i=1}^{N_d}\sum_{m=1}^{N}|f_{i,k}^*\lambda_{m}|^2}\ ,
        \end{align}
        holds.
    \end{thm}
    \begin{proof}
        The proof is similar to the proof of Theorem~\ref{thm:Shot noise error for evaluating A}, but the standard deviation of the matrix elements is bounded by Eq.~\eqref{1444}.
    \end{proof}
    \begin{corollary}
    \label{corollary: shot noise}
    The result of the Theorems~\ref{thm:Shot noise error for evaluating A} and \ref{thm:Shot noise error for evaluating C} are valid for all matrices $A$ and vectors $C$ that are calculated with the circuits in Fig.~\ref{fig:circuit A}, in particular for all possible input parameter $\bm{\theta}$ of the Ansatz in Eq.~\eqref{eq:Ansatz}. Therefore, we can state that with probability of at least $1-\eta$ the following bounds hold for all $\bm{\theta}$:
    \begin{align}\label{eq: lower bound for epsilon A}
          \|A(\bm{\theta})-\hat{A}(\bm{\theta})\|&\leq\frac{\|\{\sigma_{k,l}\}_{k,l=1}^{N_V}\|}{\sqrt{N_r  \eta}}\ ,
    \end{align}
    and
    \begin{align}
     \|C(\bm{\theta})-\hat{C}(\bm{\theta})\|&\leq \frac{\|\{\sigma_{k}\}_{k=1}^{N_V}\|}{\sqrt{N_r  \eta}}\ .
    \end{align}
    \end{corollary}
    
\section{Error bound theorems}
\label{appendix: error bound theorems}
\begin{proof}[Proof of Lemma\ref{lemma:liubovs sensitivity}:]
    Let us identify $\hat{f}$ with $f(\xi)$, the vector under disturbance $\xi$. Accordingly, $f(0)=f$. Assume that the derivative $\frac{\partial f(\xi)}{\partial \xi}$ exists and calculate the derivative of Eq.~\eqref{eq:disturbanceinthm} with respect to $\xi$:
    \begin{align}
          (A+\xi R)\frac{\partial f(\xi)}{\partial \xi}+Rf(\xi)=r  
    \end{align}
At $\xi=0$, we have:
    \begin{align}
         \frac{\partial f(\xi)}{\partial \xi}\Big|_{\xi=0}=A^{-1}(r-Rf(0)) \ . 
    \end{align}
    The Taylor expansion of $f(\xi)$ around $\xi=0$ reads
    \begin{align}
        f(\xi)&=f(0)+\xi\frac{\partial f(\xi)}{\partial \xi}\Big|_{\xi=0}+\mathcal{O}(\xi^2)\\
        &=f(0)+\xi A^{-1}(r-Rf(0))+\mathcal{O}(\xi^2)
    \end{align}
    Therefore, we can estimate
\begin{align}
    \frac{\|f(\xi)-f(0)\|}{\|f(0)\|}
    &\leq \xi\frac{\|A^{-1}(r-Rf(0))\|}{\|f(0)\|}+\mathcal{O}(\xi^2)\\
    &\leq \xi\frac{\|A^{-1}\|\|r-Rf(0)\|}{\|f(0)\|}+\mathcal{O}(\xi^2)\\
     &\leq \xi\|A^{-1}\|\frac{\|r\|+\|R\|\|f(0)\|}{\|f(0)\|}+\mathcal{O}(\xi^2)\\
    &\leq \xi\|A^{-1}\|\left(\frac{\|r\|}{\|f(0)\|}+\|R\|\right)+\mathcal{O}(\xi^2)\\
    &\leq \xi\|A^{-1}\|\|A\|\left(\frac{\|r\|}{\|A\|\|f(0)\|}+\frac{\|R\|}{\|A\|}\right)+\mathcal{O}(\xi^2)\\
     &\leq \xi\|A^{-1}\|\|A\|\left(\frac{\|r\|}{\|Af(0)\|}+\frac{\|R\|}{\|A\|}\right)+\mathcal{O}(\xi^2)\\
    &\leq \xi\kappa(A)\left(\frac{\|r\|}{\|C\|}+\frac{\| R\|}{\|A\|}\right)+\mathcal{O}(\xi^2)\ .
\end{align}
\end{proof}

    \begin{proof}[Proof of Lemma~\ref{lemma: jacobian bound}:]
        \par The gradient $\nabla_{\bm{\theta}}{\phi}(\bm{\theta}(\tau))$ evaluates in the chosen circuit as
        \begin{align*}
            \frac{\partial \phi(\theta(\tau))}{\partial \theta_k}=\sum\limits_{j=1}^{N_d}\left(f_{k,j}R_{k,j}\ket{0}^{\otimes n}\bra{\phi(\theta(\tau))}+\ket{\phi(\theta(\tau))}\bra{0}^{\otimes n}R_{k,j}^{\dagger}f_{k,j}^{\dagger}\right)\ ,
        \end{align*}
         where we used   Eq.~\eqref{eq: Ansatz derivative}.
        Then the trace norm of the dot product with the parameter vector $\bm{\theta}^*(T)$ evaluates as
        \begin{align*}
            \|\nabla_{\bm{\theta}}{\phi}(\bm{\theta_0}(T))\cdot\bm{\theta}^*(T)\|_1
            &=\left\Vert\sum\limits_{k=1}^{N_V}\left[\sum\limits_{j=1}^{N_d}\left(f_{k,j}R_{k,j}\ket{0}^{\otimes n}\bra{\phi(\bm{\theta_0}(T))}+\ket{\phi(\bm{\theta_0}(T))}\bra{0}^{\otimes n}R_{k,j}^{\dagger}f_{k,j}^{\dagger}\right)\right]  \theta^*_{k}(T)\right\Vert_1\\
            &\leq \sum\limits_{k=1}^{N_V}\left\Vert\left[\sum\limits_{j=1}^{N_d}\left(f_{k,j}R_{k,j}\ket{0}^{\otimes n}\bra{\phi(\bm{\theta_0}(T))}+\ket{\phi(\bm{\theta_0}(T))}\bra{0}^{\otimes n}R_{k,j}^{\dagger}f_{k,j}^{\dagger}\right)\right]  \theta^*_{k}(T)\right\Vert_1\\
            &\leq \sum\limits_{k=1}^{N_V}\left[\sum\limits_{j=1}^{N_d}\left(|f_{k,j}|  \left\Vert R_{k,j}\ket{0}^{\otimes n}\bra{\phi(\bm{\theta_0}(T))}\right\Vert_{1}+\left\Vert \ket{\phi(\bm{\theta_0}(T))}\bra{0}^{\otimes n}R_{k,j}^{\dagger}\right\Vert_{1}  |f_{k,j}^{\dagger}|\right)\right]  |\theta^*_{k}(T)|\\
            &\leq \sum\limits_{k=1}^{N_V}\left[\sum\limits_{j=1}^{N_d}\left(|f_{k,j}|+ |f_{k,j}^{\dagger}|\right)\right]  |\theta^*_{k}(T)|=\sum\limits_{k=1}^{N_V}\left[\sum\limits_{j=1}^{N_d}2|f_{k,j}|\right]  |\theta^*_{k}(T)| \ ,
        \end{align*}
    where in the third line, we used the fact that any quantum state has trace norm at most one.
    \end{proof}

%\nocite{*}
\bibliography{mainbib}% Produces the bibliography via BibTeX.

\end{document}

%% file: circuitAinTikZ.tex
\begin{quantikz}
\lstick{ancilla\\ $\ket{0}+e^{i\zeta}\ket{1}$}&\qw&\qw& \qw &\gate{X}\qw &\ctrl{1}\qw &\gate{X} \qw & \qw&\qw &\ctrl{1}\qw&\meter{$\ket{\pm}$}\qw  \\
\lstick{register\\ $\ket{0}^n$} & \gate{R_1} &\gate{R_2}& \ \ldots\ \qw &\gate{R_{k-1}} & \gate{U_{k,i}} &\gate{R_k}\qw& \ \ldots\ \qw &\gate{R_{l-1}}\qw& \gate{U_{l,j}}\qw&\qw
\end{quantikz}